\documentclass[journal]{IEEEtran}

\usepackage{amsmath,graphicx,url,bbm,comment,amsfonts,enumitem,subcaption,lipsum}
\usepackage{color}
\usepackage[utf8]{inputenc}
\usepackage[english]{babel}

\usepackage{amsthm,soul,color}
\usepackage{eso-pic}
 
\theoremstyle{definition}
\newtheorem{definition}{Definition}
\newtheorem{theorem}{Theorem}
\usepackage[square, comma, sort&compress,numbers]{natbib} 

\ifCLASSINFOpdf
\else
\fi
\begin{document}
 \AddToShipoutPicture*{\small \sffamily\raisebox{27.0cm}{\hspace{2.7cm}
 Published in IEEE Transactions on Image Processing (DOI: 10.1109/TIP.2021.3060875).  Copyright \copyright \hspace{1mm} 2021 IEEE.}}
%
\title{Pareto-Optimal Bit Allocation for Collaborative Intelligence}
%
%
%
\author{Saeed Ranjbar~Alvar,~\IEEEmembership{Student Member,~IEEE,}
        and Ivan V. Baji\'c,~\IEEEmembership{Senior Member,~IEEE}
        \thanks{The authors are with the School of Engineering Science, Simon Fraser University, Burnaby, BC, Canada. E-mail: saeedr@sfu.ca, ibajic@ensc.sfu.ca}}
\maketitle

\begin{abstract}
In recent studies, collaborative intelligence (CI) has emerged as a promising framework for deployment of Artificial Intelligence (AI)-based services on mobile/edge devices. In CI, the AI model (a deep neural network) is split between the edge and the cloud, and intermediate features are sent from the edge sub-model to the cloud sub-model.  
In this paper, we study bit allocation for feature coding in multi-stream CI  systems. We model task distortion as a function of rate using convex surfaces similar to those found in distortion-rate theory. Using such models, we are able to provide closed-form bit allocation solutions for single-task systems and scalarized multi-task systems. Moreover, we provide analytical characterization of the full Pareto set for 2-stream $k$-task systems, and bounds on the Pareto set for 3-stream 2-task systems. Analytical results are examined on a variety of DNN models from the literature to demonstrate wide applicability of the results. 
\end{abstract}


%
\IEEEpeerreviewmaketitle

\section{Introduction}
%
%
%
%
\IEEEPARstart{O}{ver} the last few years, the emergence of high-end mobile/edge devices with AI hardware opened new doors for AI-based applications ``at the edge''~\cite{ai_benchmark}. However, the limitations on battery and processing power restrict the size and sophistication of deep neural networks (DNNs) that can be deployed on edge devices. As a result, the current default approach for running high-end analytics on signals collected at the edge is to transfer the signals from the edge to the cloud, run the inference in the cloud, and then send the results back to the edge device if needed. 

In~\cite{neurosurgeon,collab_offloading,jointdnn, IoT_partition} it was shown that for many DNN models, partitioning the computations between the edge and cloud can decrease the inference latency and improve energy efficiency. 
We refer to this approach  as \emph{collaborative intelligence} (CI). 
In CI, initial layers of a DNN comprise the \emph{edge sub-model} and are deployed on the edge device. The deep feature tensors obtained from the edge sub-model are transferred to the cloud, where the remainder of the DNN, called the \emph{cloud sub-model}, completes the inference task. Besides energy and latency savings, another advantage of CI systems compared to the conventional cloud-only approach is their potential for higher privacy, because the original signals never leave the edge device -- only their latent representations (deep features) are transmitted to the cloud.   

To efficiently utilize the communication channel in CI, the feature tensors obtained from the edge sub-model need to be compressed prior to transmission to the cloud~\cite{intelligent_sensing}. For single-stream DNNs such as VGG~\cite{VGG}, only one feature tensor needs to be compressed, so the entire available bit budget is assigned to this tensor. However, for multi-stream DNNs, which include many recent high-performance models (e.g., models based on residual blocks~\cite{resnet,resnet_variant1,resnet_variant2}, models based on dense blocks~\cite{densenet,densenet_variant1,densenet_variant2}, etc.), it may be necessary to compress multiple tensors. A natural question then arises - how do we allocate bits among these tensors to achieve the best performance? This paper offers answers to the above question. 

Depending on the number of inference tasks they need to support, CI systems can be divided into single-task and multi-task systems. Bit allocation for single-task CI systems can be formulated as a single-objective optimization problem (SOP) and solved using Lagrangian methods. On the other hand, bit allocation for multi-task CI systems is a multi-objective optimization problem (MOP) with possibly many solutions under the notion of Pareto optimality~\cite{nonlinear_opt_book}. In our recent work~\cite{our_icassp}, the MOP related to multi-task bit allocation was converted to SOP through scalarization~\cite{Chong_Zak_Optim_2013}, and then solved using Lagrangian methods. In this paper, we present a more complete solution to the related SOP, including explicit non-negative rate constraints. Moreover, we analyze the Pareto-optimal set of solutions for the related MOP, which contains all scalarized solutions, and give an analytic characterization of the Pareto set in certain cases. 
In summary, the main contributions of this work are as follows:
\begin{itemize}
  \item A closed-form bit allocation solution for multi-stream single-task CI systems involving explicit non-negative rate constraints  (Section~\ref{sec:sop}).
  \item A closed-form bit allocation solution for scalarized multi-stream multi-tasks CI systems with explicit non-negative rate constraints (Section~\ref{sec:mop}).
  \item Analytic characterization of the Pareto set for 2-stream multi-task CI systems (Section~\ref{sec:Pareto_2xk}). 
  \item Bounds on the Pareto set for 3-stream 2-task CI systems (Section~\ref{sec:Pareto_3x2}).
  \item Experimental validation of the presented solutions on several recent DNN models (Section~\ref{sec:experiments})
\end{itemize}


The paper is structured as follows. Recent work related to CI and deep feature compression is reviewed in Section~\ref{sec:related_work}. Section~\ref{sec:preliminaries} introduces several concepts, such as the notion of task distortion and the distortion model, that will be used later in the paper. Section~\ref{sec:sop} studies bit allocation for multi-stream single task CI systems, while Section~\ref{sec:mop} analyzes multi-stream multi-task systems.  Section~\ref{sec:experiments} presents experimental validation of the derived solution, followed by conclusions in Section~\ref{sec:conclusion}.
\section{Related work}
\label{sec:related_work}
Most of the existing work on deep feature compression is focused on single-stream systems, where one feature tensor needs to be encoded. In~\cite{choi_icip}, High Efficiency Video Coding (HEVC) Range extension (RExt)~\cite{HEVC_rext} is used to encode the the intermediate feature tensor from the YOLOv2 object detector~\cite{YOLO9000}. A fine turning step called ``compression-augmented training'' is also introduced 
to mitigate the drop of performance caused by intermediate feature compression. 
A lightweight lossy codec based on entropy-constrained quantization of the deep feature is proposed in~\cite{lightweight}. In~\cite{choi_lossless}, a near-lossless compression of YOLOv2 features is presented. 
Lossless mode of HEVC~\cite{HEVC} is used to encode the feature tensors that are organized into a tiled image. The authors also introduced the idea of a reconstruction branch, which is able to reconstruct an approximation to the original image from compressed features. 
In~\cite{saeed_icip}, a loss function encouraging feature compressibility is developed and used as a part of the overall training loss for a multi-task model, resulting in more compressible features.

Authors in~\cite{intelligent_sensing,intelligent_sensing_2} describe general lossy and lossless codecs for deep feature compression. Their focus is on the features of popular DNN backbones rather than task-specific features. The authors also highlighted the need for standardization of deep feature coding and tried to provide a baseline for further research and standardization activities. In fact, a related standardization activity has recently started within the MPEG Video Coding for Machines (MPEG-VCM) ad-hoc group~\cite{MPEG-VCM_DCfE}, where the goal is to develop technologies for efficient compression of DNN intermediate features. 

While the above mentioned works exploited statistical redundancies for compression of deep feature tensors, another group of works focused on tensor dimension reduction as a means for compression~\cite{bottlenet,butterfly,bottlenet++,BaF}. In~\cite{butterfly}, an autoencoder model named  ``butterfly unit'' is utilized for feature tensor dimension reduction. In~\cite{bottlenet}, a method similar to compression-augmented training is proposed to compensate the potential accuracy loss due to the lossy compression of deep features. In~\cite{bottlenet++} another dimension reduction unit is developed based on joint source-channel coding. The back-and-forth prediction method in~\cite{BaF} enables the encoding of only a subset of tensor channels, from which the other, non-coded channels can be reconstructed.  

Due to their focus on single tensor compression, none of the studies mentioned above consider optimal bit allocation to multiple tensors. Even in~\cite{intelligent_sensing,intelligent_sensing_2} where compression of multiple features is considered, the compression is performed without joint bit allocation. The main contribution of the present paper are the solutions to bit allocation problems in several multi-stream CI scenarios. 
It should be noted that bit allocation for DNN compression has recently been studied in~\cite{Girod_ICIP2019}, where the authors propose a strategy for allocating bits to both weights and activations of a DNN, with a focus on single-stream, single-task DNNs. The present paper differs from~\cite{Girod_ICIP2019} in several ways: (1) our proposed approach supports multi-stream single-task and multi-stream multi-task models; (2) we focus on bit allocation among multiple feature tensors (i.e., activations), while the weights of the original DNN are unchanged; (3) we provide a convex approximation to the model's distortion-rate surface, which allows us to find closed-form solutions and analytically characterize the Pareto set for the bit allocation problem, unlike~\cite{Girod_ICIP2019}, where the solution is found by numerical search.
 
It should also be noted that multi-objective optimization for scalable video coding has been studied in~\cite{sc_pareto}. In particular, the authors attempt to characterize the \emph{Pareto front} in the distortion space, which is the set of all distortion points that can be achieved by Pareto-optimal bit allocations. Apart from the obvious fact that our framework considers multi-stream multi-task models for CI, rather than video coding, another important difference between~\cite{sc_pareto} and this work is that we characterize the \emph{Pareto set}, i.e, the set of Pareto-optimal rates, in certain cases. 
\section{Preliminaries}
\label{sec:preliminaries}
\subsection{Classification of CI systems}
Depending on the number of tasks it is meant to support, the architecture of the underlying DNN, and the split point between the edge sub-model and the cloud sub-model, a CI system will fall into one of the following categories.
\begin{itemize} 
\item \textbf{Single-stream, single-task systems:} one feature tensor is  transferred from the edge to the cloud to support one inference task, as shown in Fig.~\ref{fig:CNNs}(a); an example of a DNN that would fit into this scenario is VGG~\cite{VGG}. 
\item \textbf{Single-stream, multi-task systems:} one feature tensor is transferred from the edge  to the cloud to support multiple tasks, as shown in Fig.~\ref{fig:CNNs}(b); examples of DNNs for such systems are~\cite{saeed_icip, NeurIPS_Sener}. 
\item \textbf{Multi-stream, single-task systems:}  multiple tensors are transferred from the edge to the cloud to support a single task, as shown in Fig.~\ref{fig:CNNs}(c); an example of a DNN for such systems is DenseNet~\cite{densenet}. 
\item \textbf{Multi-stream, multi-task systems:} multiple tensors are transferred from the edge to the cloud to support  multiple tasks, as shown in Fig.~\ref{fig:CNNs}(d); an example of a DNN for such systems is~\cite{branched}. 
\end{itemize}

\begin{figure}[]
\begin{tabular}{c}
\begin{subfigure}{\columnwidth}\centering\includegraphics[width=7.4cm]{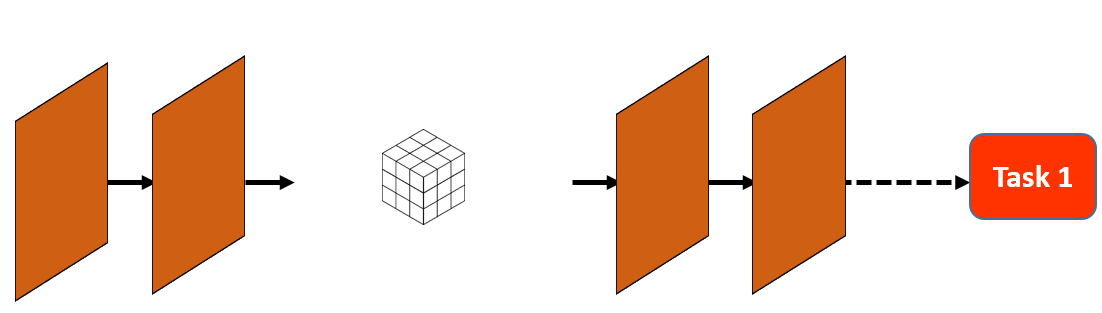}\caption{}\label{}\end{subfigure} \\
\begin{subfigure}{\linewidth}\centering\includegraphics[width=7.4cm]{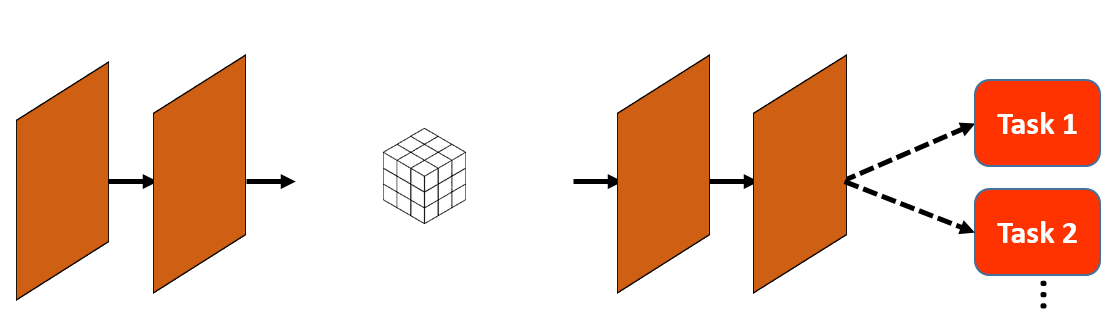}\caption{}\label{}\end{subfigure} \\
\begin{subfigure}{\columnwidth}\centering\includegraphics[width=7.4cm]{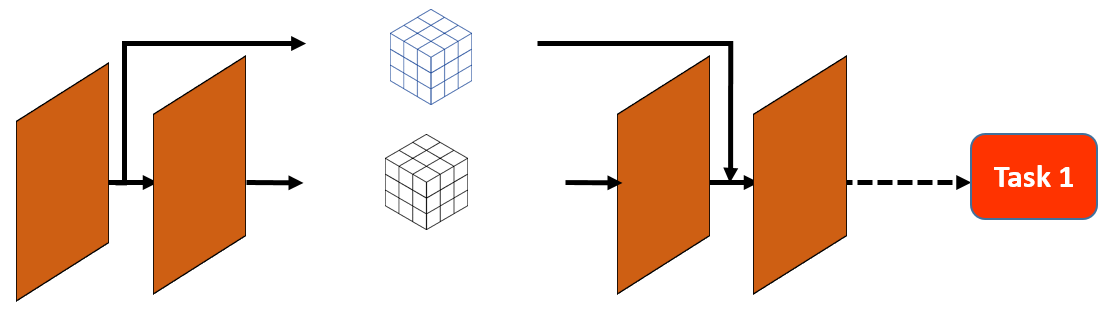}\caption{}\label{}\end{subfigure}  \\
\begin{subfigure}{\linewidth}\centering\includegraphics[width=7.4cm]{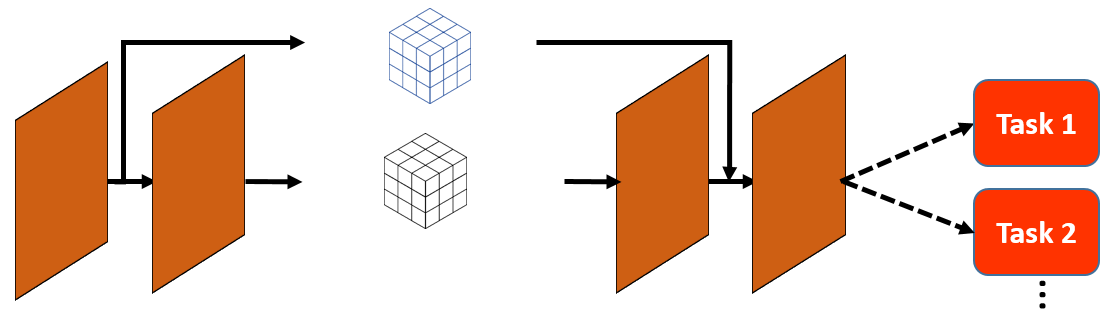}\caption{}\label{}\end{subfigure} \\
\end{tabular}
\caption{CI systems: (a) single-stream single-task; (b) single-stream multi-task; (c) multi-stream single-task; (d) multi-stream multi-task.}
\label{fig:CNNs}
\end{figure}

In single-stream systems, only one tensor is transferred from the edge to the cloud, so the entire bit budget is allocated to that tensor. However, in multi-stream systems, two or more tensors are transferred, so the crucial problem is how to allocate bits to these tensors in order to maximize the CI system performance. Not all tensors are equally important for the accuracy of the inference task(s), and the bits should be allocated accordingly. In order to be able to measure the impact of tensor compression on task accuracy, we first define task distortion and then propose a distortion-rate model.  

\subsection{Task distortion}
\label{subsec:distortion}
CI systems need to support a number of different inference tasks. In some cases, task-specific accuracy is quantified using a measure that increases with accuracy, for example classification accuracy for image classification~\cite{VGG}, mean Average Precision (mAP) for object detection~\cite{YOLO9000}, and so on. In other cases, a measure that decreases with accuracy is used, for example Root Mean Squared Error (RMSE) for disparity estimation~\cite{metrics}. These accuracy measures also have different scales, making their comparison difficult. It is therefore important to define task distortion that would map all these measures to a set of more comparable values and ensure uniform behavior (either increasing or decreasing) as the accuracy degrades. Below we define one such task distortion.

Consider a DNN model with $k\geq1$ tasks. 
Let $\overline{A_i}$ be the the model's average performance on the $i$-th task, on a given dataset, without tensor compression. We define the task-specific distortion as the fraction of the performance drop relative to the case where no compression is applied to the feature tensors. Let $A_i$ be the average performance with tensor compression on the same dataset. Then the distortion for task $i$ is defined as
\begin{equation}
\label{eq:dist_def}
    D_i = \frac{|\overline{A_i} - A_i|}{\overline{A_i}}\times100.
\end{equation}
Note that $D_i = 0$ if the accuracy with compression ($A_i$) matches the accuracy without compression ($\overline{A_i}$), and increases as $A_i$ starts to deviate from $\overline{A_i}$. $D_i$ can be interpreted as a percentage drop in performance due to feature compression. 

\subsection{Distortion-rate model}
Let $\mathbf{R}=(R_1,R_2,...,R_N)$ be the vector of bit rates for the $N$ tensors to be compressed in a multi-stream CI system. We model the dependence of task distortion on these rates using monotonically-decreasing convex surfaces given by:       
\begin{equation}
\label{eq:dist_total_form}
    D_i(\mathbf{R}) = D_i(R_1,...,R_N) \approx \gamma_i + \sum_{j=1}^{N} \alpha_{i,j} 2^{-\beta_{i,j} R_{j}},
\end{equation}
where 
$\gamma_i$, $\alpha_{i,j}>0$ and $\beta_{i,j}>0$ are surface parameters. 
In our experiments, we used non-linear least squares method based on Levenberg-Marquardt algorithm \cite{scipy_fit} 
to fit the surface~(\ref{eq:dist_total_form}) to the measured distortion-rate points. 

There are several reasons for using such a distortion-rate model. First, the model is quite accurate in approximating measured distortion-rate points. As an example, Fig.~\ref{fig:surface_class} shows a fitted surface for a single-task model (DenseNet~\cite{densenet}), with two tensors to be coded (hence, two rates). As seen in the figure, the agreement between the original points and the fitted surface is quite good. This is further confirmed quantitatively using the coefficient of multiple determination $R^2$~\cite{Neter_etal_1988}, which, for the surface in Fig.~\ref{fig:surface_class}, was $R^2=0.98$. Note that $0\leq R^2 \leq 1$, so  $R^2=0.98$ is quite high. In addition, the residuals (the differences between the actual points and the fitted surface) were clustered around zero, with mean residual being $2.4\times{10^{-10}}$. This, together with the high value of $R^2$, indicates that the model in~(\ref{eq:dist_total_form}) is an excellent approximation to the measured distortion-rate points. Indeed, in all test cases in our experiments we were obtaining $R^2>0.94$, with residuals centered around zero. 

\begin{figure}[t]	
	\centering
	\centerline{\includegraphics[width=\columnwidth]{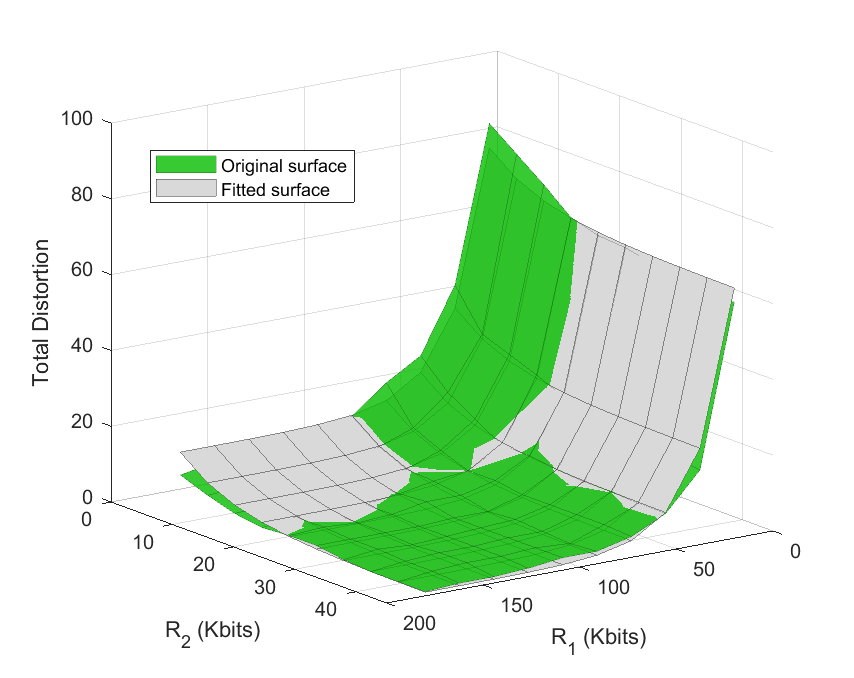}}
	\caption{distortion-rate surface obtained by encoding two deep feature tensors (green) and the fitted surface (gray). $R_1$ and $R_2$ are the average bit rates (Kbits/tensor) of the two tensors in a split DenseNet~\cite{densenet} used for image classification. 
	} 
	\label{fig:surface_class}
\end{figure} 

Another reason for selecting the model in~(\ref{eq:dist_total_form}) is the fact that theoretical distortion-rate functions~\cite{elements_IT} for commonly-used source models, such as Gaussian source with squared-error distortion and Laplacian source with absolute-error distortion, have this form, where distortion decays exponentially with increasing rate. And finally, the fact that distortion-rate surfaces in~(\ref{eq:dist_total_form}) are convex and monotonically decreasing allows us to obtain closed-form solutions for single-task and scalarized multi-task systems, as well as analytically characterize the Pareto-optimal set of solutions in certain multi-task cases.  

\section{Multi-stream single-task systems}
\label{sec:sop}

In this section we consider multi-stream single-task CI systems (Fig.~\ref{fig:CNNs}(c)), where multiple tensors are transferred from the edge to the cloud to support a single inference task. 
Let $\mathbf{R}=(R_1,R_2,...,R_N)$ be the vector of bit rates assigned to the encoded tensors. The distortion of the single task, with tensors encoded at rates $\mathbf{R}$, is $D(\mathbf{R})$, and is modeled using~(\ref{eq:dist_total_form}). 
In the single-task case, index $i$ can be dropped from~(\ref{eq:dist_total_form}), so the model  becomes $D(\mathbf{R}) \approx \gamma + \sum_{j=1}^{N} \alpha_{j} 2^{-\beta_{j} R_{j}}$, 
where $\gamma$, $\alpha_{j}$ and $\beta_{j}$ are the estimated surface parameters.
The goal is to pick $\mathbf{R}$ that minimizes the distortion subject to the total rate constraint and non-negative rates:
\begin{equation}
\label{eq:rdo}
\begin{gathered}
\text{minimize} \quad D(\mathbf{R}) \\
\begin{aligned}
\text{ s. t.} \quad \mathbf{R} \in S =  \bigg\{\mathbf{R}\in \mathbb{R}^N : &\sum_{j=1}^{N} R_j \leq R_t \\
&R_j \geq 0, \enspace j=1,...,N\bigg\}. \end{aligned}
\end{gathered}
\end{equation}
Since $D(\mathbf{R})$ and the feasible set $S$ are convex, problem~(\ref{eq:rdo}) has a closed-form solution, which 
can be obtained using the standard method of Lagrange multipliers~\cite{Boyd_Vandenberghe_2004}. Specifically, the constrained problem in~(\ref{eq:rdo}) is converted to an unconstrained problem of minimizing the Lagrangian $J(\mathbf{R})$, given by
\begin{equation}
\label{eq:lagrangian_water}
    J(\mathbf{R}) = D_t(\mathbf{R}) + \lambda \cdot \left(\sum_{j=1}^{N} R_j - R_t\right) - \sum_{j=1}^{N} \mu_j R_j  
\end{equation}
where $\lambda$ and $\mu_j$ are the Lagrange multipliers. According the the Karush-Kuhn-Tucker (KKT) conditions~\cite{Nonlinear_prog}, a point $\mathbf{R}^*=(R_1^*, ..., R_N^*)$ is the solution of~(\ref{eq:rdo}) if the following conditions are satisfied for each $j=1,2,...,N$:  
\begin{equation}
\label{eq:stationarity}
    \frac{\partial J(\mathbf{R}^*)}{\partial R_j}=0,
\end{equation}
\begin{equation}
\label{eq:dual_feasibility}
    \lambda \geq 0, \quad \mu_j \geq 0,
\end{equation}
\begin{equation}
\label{eq:comp_slackness}
    \mu_jR_j = 0.
\end{equation}
From~(\ref{eq:stationarity}) we obtain for each $j=1,2,...,N$, 
\begin{equation}
    (\ln{2}) (-\alpha_j\beta_j) 2^{-\beta_j R_j^*}+\lambda-\mu_j=0.
    \label{eq:zero_derivative}
\end{equation}
Since $R_j^*\geq0$, we have either $R_j^*=0$ or $R_j^*>0$. For those $j$ for which $R_j^*=0$, from equation~(\ref{eq:zero_derivative}) we get:
\begin{equation}
     \mu_j = \lambda - (\ln{2})(\alpha_j\beta_j),
\end{equation}
and because $\mu_j\geq0$, we conclude that 
\begin{equation}
\lambda \geq  (\ln{2})(\alpha_j\beta_j).
\label{eq:zero_rate}
\end{equation}
Hence, condition~(\ref{eq:zero_rate}) is associated with $R^*_j=0$. On the other hand, for those $j$ for which $R_j^*>0$, from~(\ref{eq:comp_slackness}) we must have $\mu_j=0$. Then, from~(\ref{eq:zero_derivative}):
\begin{equation}
  \lambda =  (\ln{2})(\alpha_j\beta_j) 2^{-\beta_j R_j^*}
\end{equation}
and solving for $R_j^*$ we obtain:
\begin{equation}
    R_j^* = \frac{1}{\beta_j} \bigg[\log_2\big\{(\ln{2})(\alpha_j\beta_j)\big\}-\log_2 \lambda \bigg].
\end{equation}
To find $\lambda$, we first note that the optimal solution must satisfy $\sum_{j=1}^{N} R_j^* = R_t$. If this were not the case and we had $\sum_{j=1}^{N} R_j^* < R_t$, we could increase any of the $R_j$ and the distortion in~(\ref{eq:dist_total_form}) would be reduced. Hence, using the fact that in the optimum solution, the rate constraint is satisfied with equality, we get
\begin{equation}
    \sum_{j \, : \, R^*_j>0} \frac{1}{\beta_j} \bigg[\log_2\big\{(\ln{2})(\alpha_j\beta_j)\big\}-\log_2 \lambda \bigg] = R_t,
\end{equation}
from which we get
\begin{equation}
\log_2 \lambda = \frac{\left[\sum_{j \, : \, R^*_j>0}{\frac{1}{\beta_j}\log_2\big\{(\ln{2})(\alpha_j\beta_j)\big\}}\right] - R_t}{\sum_{j \, : \, R^*_j>0}{\frac{1}{\beta_j}}}.
\end{equation}

Finally, we note that condition~(\ref{eq:zero_rate}), which implies $R^*_j=0$, is equivalent to 
$\big[\log_2\big\{(\ln{2})(\alpha_j\beta_j)\big\}-\log_2 \lambda \big] \leq 0$. We can therefore express the optimal rates compactly as
\begin{equation}
    R_j^* = \frac{1}{{\beta_j}} \bigg[\log_2\big\{(\ln{2})(\alpha_j\beta_j)\big\}-\log_2 \lambda \bigg]^+,
    \label{eq:solution_SOP}
\end{equation}
for $j=1,2,...,N$, where $[x]^+ = \max(x,0)$. This type of solution is sometimes called ``reverse water-filling''~\cite{elements_IT}, with $\log_2 \lambda$ being the ``water level.'' Those $j$ for which $(\ln{2})(\alpha_j\beta_j)>\lambda$ get a positive rate ($R_j^*>0$), and others get zero rate ($R_j^*=0$). Solution~(\ref{eq:solution_SOP}) will also be useful for scalarized multi-task systems, which we will discuss in the next section.

\section{Multi-stream multi-task systems}
\label{sec:mop}
In a multi-stream multi-task CI system (Fig.~\ref{fig:CNNs}(d)), multiple tensors need to be transferred from the edge to the cloud, where they will be used to support $k>1$ tasks. As before, let  $\mathbf{R}=(R_1,R_2,...,R_N)$ be the vector of bit rates assigned to the encoded tensors, and let $D_i(\mathbf{R})$ be the distortion of the $i$-th task, $i=1,2,...,k$.
Bit allocation for a multi-task system is a multi-objective optimization problem (MOP):
\begin{equation}
\label{eq:rdo_mop}
\begin{gathered}
\text{minimize} \quad \big\{D_1(\mathbf{R}), D_2(\mathbf{R}), ..., D_k(\mathbf{R})\big\} \\
\begin{aligned}
\text{ s. t.} \quad \mathbf{R} \in S =  \bigg\{\mathbf{R}\in \mathbb{R}^N : &\sum_{j=1}^{N} R_j \leq R_t \\
&R_j \geq 0, \enspace j=1,...,N\bigg\}. \end{aligned}
\end{gathered}
\end{equation}
Here, the feasible set $S$ is defined in the same as in the single-task problem~(\ref{eq:rdo}). However, there are now multiple cost functions (distortions) that need to be minimized. In most practical MOPs, improving one cost often causes other costs to degrade, and this will also be seen in our experiments in Section~\ref{sec:experiments}. In order to capture the notion of optimality in such cases where costs are conflicting, Pareto optimality is defined~\cite{nonlinear_opt_book}.
\begin{definition}
\label{def:pareto}
A vector $\mathbf{R}^* \in S$ is \emph{Pareto optimal} if there does not exist another vector $\mathbf{R} \in S$ such that $D_i(\mathbf{R}) \leq D_i(\mathbf{R}^*)$ for all $i=1,\dots,k$ and $D_j(\mathbf{R}) < D_j(\mathbf{R}^*)$ for at least one index $j$.  
\end{definition}
In other words, for the MOP in (\ref{eq:rdo_mop}), a solution is called Pareto optimal if none of the distortions can be improved without deteriorating at least one other distortion. The set of all Pareto-optimal solutions $\mathbf{R}^*$ is called the \emph{Pareto set}, and the corresponding set of cost function values $D_i(\mathbf{R}^*)$ on the Pareto set is called the \emph{Pareto front}. There is also a weaker notion of optimality for a MOP, called weak Pareto optimality~\cite{nonlinear_opt_book}. At the weakly Pareto-optimal points, it is possible to improve some cost functions (distortions) without penalizing others, i.e., with other costs remaining the same.
\begin{definition}
\label{def:weak_pareto}
A vector $\mathbf{R}^* \in S$ is \emph{weakly Pareto optimal} if there does not exist another vector $\mathbf{R} \in S$ such that $D_i(\mathbf{R}) < D_i(\mathbf{R}^*)$ for all $i=1,\dots,k$.
\end{definition}
Similarly to the single-objective optimization in~(\ref{eq:rdo}), our MOP~(\ref{eq:rdo_mop}) involves monotonically decreasing distortions. Therefore, any Pareto-optimal (or weakly Pareto-optimal) solution $\mathbf{R}^*$ must satisfy the rate constraint with equality, i.e., $\sum_{j=1}^{N} R^*_j = R_t$. Otherwise, if we had $\sum_{j=1}^{N} R^*_j < R_t$, we could simply split the leftover bits among all rates and increase each $R_j$ by the amount $\left(R_t - \sum_{j=1}^{N} R^*_j\right)/N$, thereby achieving simultaneous decrease of all  $D_i(\mathbf{R}), i=1,\dots,k$. 


Pareto set is infinite, in general. Numerical methods, such as the genetic algorithm (GA)-based methods~\cite{GA}, can be used to sample the Pareto set of~\eqref{eq:rdo_mop} to obtain a finite number of Pareto-optimal solutions. However, due to the convexity of the feasible set $S$ and distortions in~(\ref{eq:rdo_mop}), our problem is a \emph{convex MOP}. This allows us give explicit solutions to~(\ref{eq:rdo_mop}) in certain cases, and give analytic characterization of the Pareto set for some CI systems. We first recall an important theorem from~\cite{nonlinear_opt_book} concerning convex MOPs. 


\begin{theorem}
\label{th:weight_positive_pareto_inverse}
For a convex MOP, let $\mathbf{w}=(w_1,\dots,w_k)$ such that $w_i> 0$ and $\sum^k_{i=1}{w_i} = 1$. If  $\mathbf{R}^*$ is a solution of the following problem:
\begin{equation}
   \text{minimize } \; D_t(\mathbf{R}) = \sum_{i=1}^{k} w_iD_i(\mathbf{R}), \quad
\text{s. t. } \; \mathbf{R} \in S,
\label{eq:scalarization}
\end{equation}
then $\mathbf{R}^*$ is also Pareto-optimal for the corresponding MOP.
\end{theorem}

Taking a weighted linear combination of cost functions in~(\ref{eq:scalarization}) is sometimes referred to as \emph{scalarization} because it creates a single (scalar) cost function out of many cost functions in the original MOP, and converts the MOP to a single-objective optimization problem. Each weight can be interpreted as the relative importance of the corresponding cost. The above theorem states that positive weights $w_i$ in~(\ref{eq:scalarization}) are a sufficient condition to obtain a Pareto-optimal solution to the MOP. However, this is not a necessary condition, and there could be Pareto-optimal solutions with zero weights. Examples of such Pareto-optimal solutions with zero weights are the endpoints of the Pareto set in our experiments in Section~\ref{exp:Multi-task}. 

Using Theorem~\ref{th:weight_positive_pareto_inverse}, there are two ways to obtain a Pareto-optimal solution for~(\ref{eq:rdo_mop}) for a given vector of task weights $\mathbf{w}=(w_1,\dots,w_k)$. In the first approach, which we call \emph{scalarize-first}, each task distortion is measured at a chosen set of rates $\mathbf{R}$. Then for each $\mathbf{R}$, a weighted combination is formed as in~(\ref{eq:scalarization}) to compute the total distortion $D_t(\mathbf{R})$. The surface model $D_t(\mathbf{R}) \approx \gamma + \sum_{j=1}^{N} \alpha_{j} 2^{-\beta_{j} R_{j}}$ is then fitted to the computed total distortions, and we end up with a single-objective optimization problem with the same structure as~(\ref{eq:rdo}), whose solution is given by~(\ref{eq:solution_SOP}). 

The second approach, which we call \emph{fit-first}, is to measure task distortions at a chosen set of rates $\mathbf{R}$, and then fit surfaces~(\ref{eq:dist_total_form}) to each task distortion $D_i(\mathbf{R})$, $i=1,\dots,k$. After that, the total distortion $D_t(\mathbf{R})$ is formed as a weighted linear combination of the fitted surfaces, as in~(\ref{eq:scalarization}). In this case, the total distortion $D_t(\mathbf{R})$ no longer has the form $\gamma + \sum_{j=1}^{N} \alpha_{j} 2^{-\beta_{j} R_{j}}$, and there is no closed-form solution. Nonetheless, $D_t(\mathbf{R})$ is still a convex surface, being a non-negative linear combination of convex surfaces $D_i(\mathbf{R})$, so the resulting problem can be solved numerically.

Ideally, for a given weight vector $\mathbf{w}$, both scalarize-first and fit-first approach should give the same solution. In practice, however, the solutions may differ somewhat due to the nature of numerical optimization. In the scalarize-first approach, the parameters of the total distortion surface $D_t(\mathbf{R})$ are numerically fitted to the weighted sum of measured task distortions. In the fit-first approach, parameters of each task distortion surface $D_i(\mathbf{R})$ are numerically fitted to the measured task distortions, and then their weighted sum is used as a cost function in the numerical solution to~(\ref{eq:scalarization}). Each of these numerical procedures may introduce some error in the computed quantities, depending on the stopping criteria and other parameters. As an example, Table~\ref{tbl:scalarize_fit} shows several rate pairs computed using the two approaches for the $2\times 3$ CI system used in the experiments in Section~\ref{exp:Multi-task}, for two weight vectors. As seen in the table, the computed rates are close, but not equal.  

\begin{table}[]
\caption{Comparison of rate pairs $(R_1,R_2)$ obtained by \emph{scalarize-first} and \emph{fit-first} approaches for the $2\times 3$ CI system used in the experiments in Section~\ref{exp:Multi-task}.}
    \centering
    \begin{tabular}{|c|c|c|}
    \hline
    $R_t=1000$ & \begin{tabular}[c]{@{}c@{}} scalarize-first\end{tabular} & \begin{tabular}[c]{@{}c@{}} fit-first\end{tabular}   \\ \hline \hline \\[-1em]
    \begin{tabular}[c]{@{}c@{}}
    $\mathbf{w}= (\frac{1}{3}, \frac{1}{3}, \frac{1}{3})$\end{tabular} & (723.28, 276.72) &  (728.73, 271.27)   \\ [-1em]
    \\
    \hline
    \begin{tabular}[c]{@{}c@{}}
    \\[-1em]
    $\mathbf{w}= (0.5, 0.48, 0.02)$\end{tabular} & (601.23, 398.77)  & (607.48, 392.52)   \\ [-1em]
    \\ \hline
    \end{tabular}%
\label{tbl:scalarize_fit}
\end{table}

Both scalarize-first and fit-first approaches give one Pareto-optimal solution of~(\ref{eq:rdo_mop}) for one weight vector $\mathbf{w}$. By changing $\mathbf{w}$, one can sample the Pareto set of~(\ref{eq:rdo_mop}) at discrete points. Other numerical Pareto solvers, such as~\cite{GA}, also offer the possibility to obtain a discrete set of samples from the Pareto set for a given MOP. However, none of these methods is able to characterize the complete Pareto set. In this paper, due to convexity of our MOP, we are able to go deeper in certain cases. Specifically, we provide analytical characterization of the Pareto set for 2-stream $k$-task ($2\times k$) CI systems, and bounds on the Pareto set for 3-stream 2-task ($3\times2$) CI systems.

\subsection{Pareto set for 2-stream $k$-task CI systems}
\label{sec:Pareto_2xk}
Consider a $2\times k$ CI system, where two feature tensors are transferred from the edge to the cloud to support $k>1$ inference tasks. 
We already know that in the optimal solution, $R_1+R_2=R_t$. Let $\mathbf{R}^i=(R_1^i,R_2^i)$ be the rates that minimize the $i$-th task distortion $D_i(\mathbf{R})$ along the rate constraint line $R_1+R_2=R_t$. Since $D_i(\mathbf{R})$ is convex, its intersection with the plane defined by $R_1+R_2=R_t$ is also convex, so $\mathbf{R}^i$ is a unique minimum of $D_i(\mathbf{R})$ along the rate constraint line.  Using~(\ref{eq:dist_total_form}) and the fact that the rates add up to $R_t$, we obtain:
\begin{equation}
\begin{gathered}
\label{eq:min_each_dist}
R_1^i = \frac{\log_2(\alpha_{i,1} \beta_{i,1})-\log_2(\alpha_{i,2} \beta_{i,2})+ \beta_{i,2} R_t}{\beta_{i,1}+\beta_{i,2}},\\
R_2^i = R_t - R_1^i. 
\end{gathered}
\end{equation}
Performing this for each of the $k$ tasks, we end up with $k$ points $\mathbf{R}^1,\mathbf{R}^2,...,\mathbf{R}^k$ on the line $R_1+R_2=R_t$. Let $R_1^{\min}=\min\{R_1^1,\dots,R_1^k\}$, $R_1^{\max}=\max\{R_1^1,\dots,R_1^k\}$, $R_2^{\min}= R_t - R_1^{\max}$ and $R_2^{\max}= R_t - R_1^{\min}$. Then we have the following result.

\begin{theorem}
\label{th:2stream_pareto}
For a 2-stream $k$-task CI  system, any point on the line $R_1+R_2=R_t$ between $(R_1^{\min},R_2^{\max})$ and $(R_1^{\max},R_2^{\min})$ is Pareto-optimal. Moreover, there are no Pareto-optimal solutions outside of this line segment.
\end{theorem}
\begin{proof}
See Appendix~\ref{appendix:2stream_proof}.
\end{proof}

An illustration of Theorem \ref{th:2stream_pareto} is given in Fig.~\ref{fig:m4_2stream_proof} for a $2 \times 3$ CI system.  
There are three task distortions, and the three convex curves (blue, yellow, and red) represent the intersections of the three distortion surfaces with the plane defined by $R_1+R_2=R_t$. The green line segment is the Pareto set stated in Theorem~\ref{th:2stream_pareto}. If a point is chosen outside of the green line segment, then moving towards the line segment would reduce all distortions simultaneously. Therefore, points outside of this line segment cannot be Pareto-optimal. 
For each point inside the green line segment, moving in either direction will increase some distortions and reduce the others, however, not all distortions can be reduced simultaneously. To see this, note that in Fig.~\ref{fig:m4_2stream_proof}, within the green line segment, moving towards the red X reduces the distortion corresponding to the red curve, but increases the distortion corresponding to the blue curve. On the other hand, moving towards the blue X reduces the distortion corresponding to the blue curve but increases the distortion corresponding to the red curve. 

\begin{figure}[t!]	
	\centering
	\centerline{\includegraphics[scale=0.71]{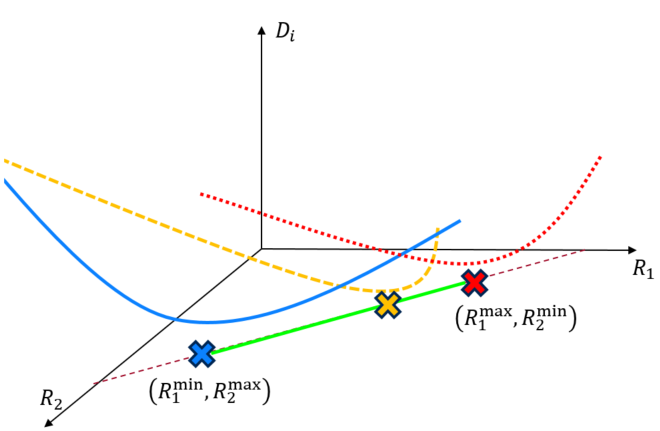}}
	\caption{An illustration of Theorem~\ref{th:2stream_pareto} for a $2 \times 3$ CI system. The dashed line represents $R_1+R_2=R_t$ and the green line segment is the Pareto set. Each curve is the intersection of $D_i$ with the plane defined by $R_1+R_2=R_t$. The three X points are the minima of each $D_i$ along $R_1+R_2=R_t$. } 
	\label{fig:m4_2stream_proof}
\end{figure} 

\subsection{Pareto set for 3-stream 2-task CI systems}
\label{sec:Pareto_3x2}
As the number of tensors to be transferred to the cloud increases beyond two, it becomes progressively more complicated to characterize the complete Pareto set. However, in the special case of $3 \times 2$ CI systems (three tensors, two tasks) we are able to provide explicit bounds on the Pareto set. We first recall an important result from~\cite{nonlinear_opt_book} concerning Karush-Kuhn-Tucker (KKT)  conditions for a convex MOP.
\begin{theorem}
\label{th:kkt}
Consider a MOP in~\eqref{eq:rdo_mop}, with convex and continuously differentiable cost and constraint functions. Let its constraint functions be denoted as $g_1(\mathbf{R})=\sum_{j=1}^{N}R_j-R_t \leq 0$ and $g_{j+1}(\mathbf{R}) = -R_j \leq 0$, $j=1,\dots,N$. A necessary and sufficient condition for $\mathbf{R}^* \in S$ to be Pareto-optimal is that there exist vectors  $\mathbf{0} < \mathbf{w} \in \mathbb{R}^k$ and $\mathbf{0} \le \boldsymbol{\mathrm{\lambda}} \in \mathbb{R}^{N+1}$ such that
\begin{equation}
\label{eq:kkt}
\begin{gathered}
\sum^k_{i=1}{w_i}\mathrm{\nabla }D_i\left(\mathbf{R}^*\right)+ \sum^{N+1}_{j=1}{{\lambda }_j}\mathrm{\nabla }g_j\left(\mathbf{R}^*\right)=0, \\
{\lambda }_jg_j\left(\mathbf{R}^*\right)=0, \quad j=1,\dots, N+1.
\end{gathered}
\end{equation}
\end{theorem}

In the above theorem, it can be assumed that $\sum^k_{i=1}{w_i} = 1$ without loss of generality~\cite{nonlinear_opt_book}. Also, for the moment, assume that all optimal rates are positive (zero and negative rates will be handled explicitly by clipping below). Since our constraint functions are $ g_{j+1}(\mathbf{R})=-R_j$ for $j=1,\dots,N$, from the bottom row in~(\ref{eq:kkt}) we conclude $\lambda_j=0$ for $j=2,\dots,N+1$. Hence, for positive rates, solving \eqref{eq:kkt} is equivalent to solving the following system of $N+1$ nonlinear equations. 
\begin{equation}
\label{eq:mop_system}
\begin{gathered}
w_1\left[ \begin{array}{c}
\frac{\partial D_1(\mathbf{R}^*)}{\partial R_1} \\ 
\vdots  \\ 
\frac{\partial D_1(\mathbf{R}^*)}{\partial R_N} \end{array}
\right]+\dots +w_k\left[ \begin{array}{c}
\frac{\partial D_k(\mathbf{R}^*)}{\partial R_1} \\ 
\vdots  \\ 
\frac{\partial D_k(\mathbf{R}^*)}{\partial R_N} \end{array}
\right]+{\lambda}_1\left[ \begin{array}{c}
1 \\ 
\vdots  \\ 
1 \end{array}
\right]=0 \\[5pt]
\lambda_1 g_{1}(\mathbf{R}^*) = 0
\end{gathered}
\end{equation}

Due to the exponential form of  distortion surfaces in~\eqref{eq:dist_total_form}, the system in~\eqref{eq:mop_system} cannot be solved analytically. 
Nonetheless, we are still able to find explicit bounds for the Pareto set for a $3 \times 2$ CI system ($N=3, k=2$). 
That is to say, we can show where the Pareto set must be, even if we can't provide an explicit equation for it.

We start from~\eqref{eq:mop_system} with $N=3$ and move  $\lambda_1$ to the right-hand side of these equations. The first three equations become: 
\begin{equation}
\label{eq:3stream_initial}
\begin{aligned}
w_1 \alpha_{1,1} \beta_{1,1} 2^{- \beta_{1,1} R_1} + w_2 \alpha_{2,1} \beta_{2,1} 2^ {- \beta_{2,1} R_1} = \lambda_1^{\prime}, \\
w_1 \alpha_{1,2} \beta_{1,2} 2^{- \beta_{1,2} R_2} + w_2 \alpha_{2,2} \beta_{2,2} 2^ {- \beta_{2,2} R_2} = \lambda_1^{\prime}, \\
w_1 \alpha_{1,3} \beta_{1,3} 2^{- \beta_{1,3} R_3} + w_2 \alpha_{2,3} \beta_{2,3} 2^ {- \beta_{2,3} R_3} = \lambda_1^{\prime}, \\
\end{aligned}
\end{equation}
where $\lambda_1^{\prime} = - \frac{\lambda_1}{\ln{2}}$. 
Equating the first and third equation, the ratio of $\frac{w_1}{w_2}$ is obtained as:
\begin{equation}
\label{eq:w1w2_ration}
\frac{w_1}{w_2} = - \frac{\alpha_{2,1} \beta_{2,1} 2^ {- \beta_{2,1} R_1} - \alpha_{2,3} \beta_{2,3} 2^ {- \beta_{2,3} R_3}}{\alpha_{1,1} \beta_{1,1} 2^{- \beta_{1,1} R_1} - \alpha_{1,3} \beta_{1,3} 2^{- \beta_{1,3} R_3}}.
\end{equation}
Using the fact that $\frac{w_1}{w_2} >0$ (because both $w_1$ and $w_2$ are positive), extreme points of $R_3$ can be obtained as functions of $R_1$. Similarly, by equating the first and the second equation in~\eqref{eq:3stream_initial}, extreme points $R_2$ can be obtained as functions of $R_1$. Then using the fact that in the optimal solution,  $R_1+R_2+R_3 = R_t$, two extreme points of $R_1$, denoted $R_1^{\min}$ and $R_1^{\max}$, can be obtained as functions of $R_t$ and surface parameters. These are shown in equation~(\ref{eq:r1_ext}) in Appendix~\ref{appendix:extreme_3x2}. Note that the extreme values of $R_1$ are clipped to $[0,R_t]$ in~(\ref{eq:r1_ext}). Similarly, we can find extreme values for the other two rates that are also clipped to $[0,R_t]$: $R_2^{\min}$ and $R_2^{\max}$, shown in~(\ref{eq:r2_ext}); and  $R_3^{\min}$ and $R_3^{\max}$, shown in~\eqref{eq:r3_ext} in Appendix~\ref{appendix:extreme_3x2}.

The extreme rate values  in~\eqref{eq:r1_ext}--\eqref{eq:r3_ext} define the cube
\begin{equation}
\label{eq:cube}
 \mathcal{C}=\big\{(R_1,R_2,R_3)\, : \, R_j^{\min} \leq R_j \leq R_j^{\max}, \,
 j=1,2,3\big\}. 
\end{equation}
We also know the optimal solutions must lie on the plane
\begin{equation}
\label{eq:plane}
    \mathcal{P} = \big\{(R_1,R_2,R_3)\, : \, R_1+R_2+R_3=R_t \big\}.
\end{equation}
Their intersection, $\mathcal{C} \, \cap \, \mathcal{P}$, defines a region where Pareto-optimal solutions can be found. This intersection is a polygon with up to six sides, as shown in Figure~\ref{fig:cube}. Not every point inside this region is Pareto-optimal, but we are able to show that there are no Pareto-optimal solutions outside of this region. Hence, this region is a bound for the Pareto set of $3 \times 2$ CI systems.

\begin{theorem}
\label{th:3stream_pareto}
For a 3-stream 2-task CI system, there are no Pareto-optimal solutions outside of $\mathcal{C} \, \cap \, \mathcal{P}$.
\end{theorem}
\begin{proof}
See Appendix~\ref{appendix:3stream_proof}.
\end{proof}

\begin{figure}[t!]	
	\centering
	\centerline{\includegraphics[width=7.5cm]{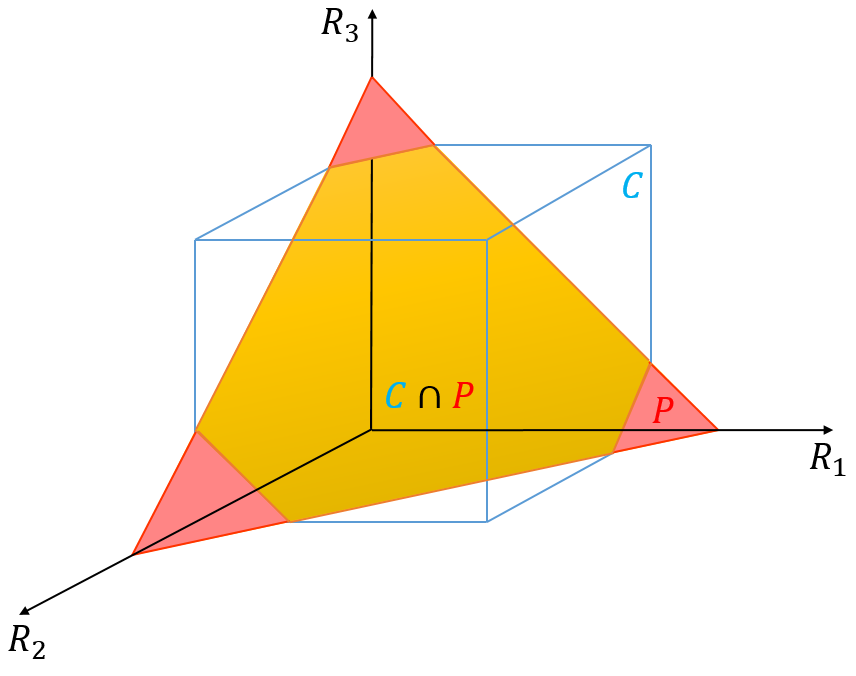}}
	\caption{An example of the intersection $\mathcal{C} \, \cap \, \mathcal{P}$. } 
	\label{fig:cube}
\end{figure}

\section{Experiments}
\label{sec:experiments}
To assess the performance of the bit allocation strategies presented earlier, experiments are conducted on various single-task and multi-task CI systems. We used DenseNet-121~\cite{densenet} for single-task experiments, and the model from~\cite{our_icassp} for multi-task experiments. Both DNN models are shown in Fig.~\ref{fig:curve}. The four test cases included in the experiments are: $2 \times 1$ (two tensors, single task), $2 \times 2$ (two tensors, two tasks), $2 \times 3$ (two tensors, three tasks) and $3 \times 2$ (three tensors, two tasks). 

\begin{figure*}[h]
\begin{tabular}{cc}
\begin{subfigure}{\linewidth}\centering\includegraphics[width=14cm]{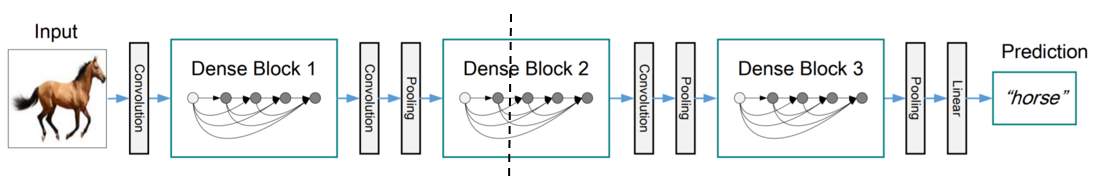}\caption{}\label{fig:dense_net}\end{subfigure} \\
\begin{subfigure}{\linewidth}\centering\includegraphics[width=10cm]{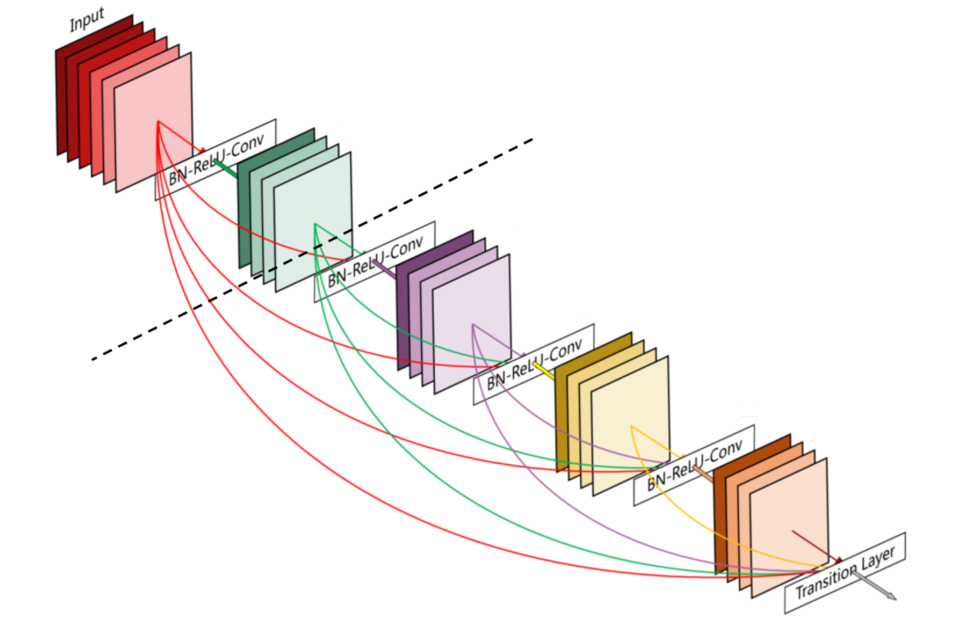}\caption{}\label{fig:dense_block}\end{subfigure} \\
\begin{subfigure}{\linewidth}\centering\includegraphics[width=16cm]{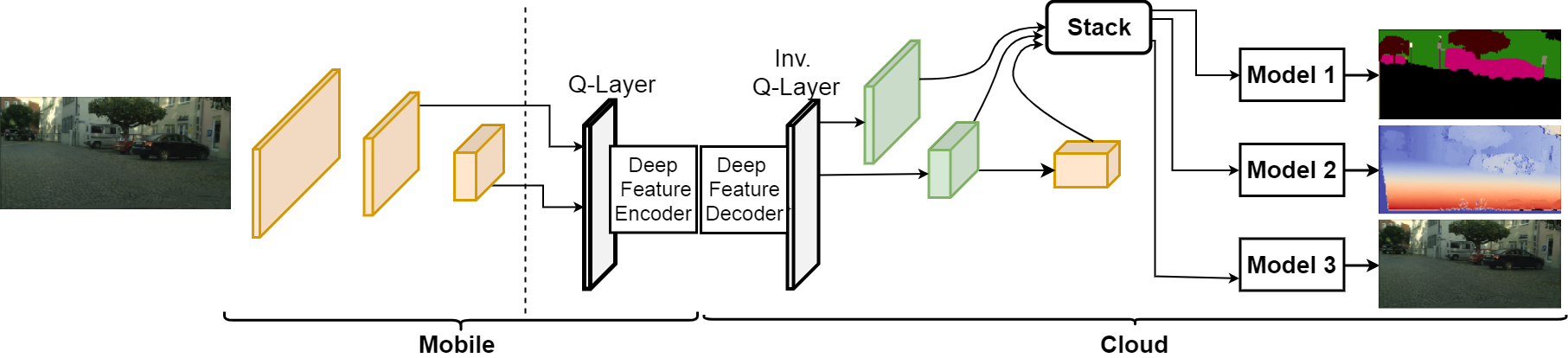}\caption{}\label{fig:multi-task_model}\end{subfigure}
\end{tabular}
\caption{DNN models used in the experiments: (a) \copyright 2017 IEEE. DenseNet~\cite{densenet}; the dashed line indicates where the model is split. (b) \copyright 2017 IEEE. Illustration of tensors in a simple dense block of DenseNet. (c) \copyright 2020 IEEE. The 2-stream, 3-task DNN from~\cite{our_icassp}. For clarity, feature tensors are shown, rather than layers. 
}
\label{fig:curve}
\end{figure*}

\subsection{Methodology}
\label{subsec:methodology}
For actual tensor compression, we follow a method similar to~\cite{choi_icip,choi_lossless}. We apply uniform 8-bit (min-max) quantization to the deep feature tensor, then rearrange its channels to form a tiled image, as shown in the sample in Fig.~\ref{fig:tiled}. Any image codec can be used to encode the tiled tensor image; in our experiments we used  JPEG2000~\cite{jpg2k} because of its rate control tools that allowed us to obtain the desired rate fairly accurately.

\begin{figure}[t!]	
	\centering
	\centerline{\includegraphics[width=\columnwidth]{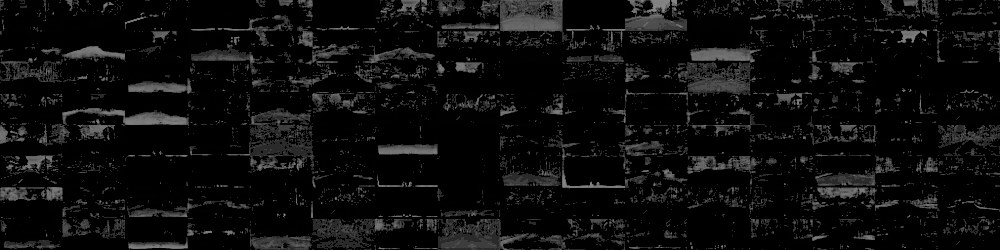}}
	\caption{An example of the tiled quantized deep feature tensor (enhanced for better visualization).} 
	\label{fig:tiled}
\end{figure} 

Each experiment covers a certain range of total rates $R_t \in [R_t^{\min}, R_t^{\max}]$. In each case, we measured task distortions at $100$ rate tuples uniformly sampled such that $0.5R_t^{\min} \leq \sum_j R_j \leq 1.5R_t^{\max}$. The reason for sampling beyond $[R_t^{\min}, R_t^{\max}]$ is to get a good distortion surface fit for cases where the actual $R_t$ is at the boundary of the range. Then, for a given $R_t$, we would select the points near $R_t$ for surface fitting. Specifically, we took all points such that $0.75R_t \leq \sum_j R_j \leq 1.25R_t$, and used these points for distortion surface fitting.

The bit allocation methods presented in Sections~\ref{sec:sop} and~\ref{sec:mop} are compared against three alternatives: equal bit allocation (Method 1), bit allocation proportional to the number of tensor elements (Method 2), and bit allocation proportional to the variance of tensor elements (Method 3). The variance for Method 3 is computed on the same training samples 
used to obtain surface parameters for our proposed bit allocation.

\subsection{Multi-stream single-task experiments}
\label{exp:single-task}
For multi-stream single-task experiments we used DenseNet-121~\cite{densenet}, a popular DNN model for image classification. As shown in Fig.~\ref{fig:curve}(a), DenseNet consists of several \emph{dense blocks}, each of which contains a number of  convolutional layers. In a dense block (Fig.~\ref{fig:curve}(b)), the initial tensor (shown in red) is passed through a \emph{dense layer}, which is a combination of batch normalization, activation function, and a convolutional layer. The input tensor is then concatenated with the output of the first dense layer, and the concatenated tensor is fed into the second dense layer, and so on. 
We used the DenseNet-121 implementation from Torchvision~\cite{densenet_model}, which was trained on ImageNet~\cite{imagenet}. No additional tuning of the model was performed. 

Due to the many skip connections in a dense block, depending on where the model is split, multiple tensors may be obtained. For our experiments, we chose a split point at the output of the second dense layer in the second dense block, thereby obtaining two tensors. 
The dimensions of the tensors are given in Table~\ref{tbl:feature_dimension_densenet}. As shown in the table, for an input image of dimensions $224\times224\times3$, the two tensors have dimensions $28\times28\times128$ and $28\times28\times32$. In total, this is 17\% less than the number of pixels in the input image.

The experiments are conducted on a subset of the ImageNet validation set 
with three total rate constraints $R_t \in \{50,100,150\}$ Kbits. We randomly selected 20 balanced classes from the validation set. In each class, 20\% of the data is randomly chosen 
to obtain distortion-rate surface (Fig.~\ref{fig:surface_class}) parameters, and the remaining 80\% of the data is used for testing. Distortion~(\ref{eq:dist_def}) is derived from the Top-1 accuracy. 
With reference to Section~\ref{subsec:methodology}, the distortion-rate surface in Fig.~\ref{fig:surface_class} is obtained by setting $R_t^{min}=50$, $R_t^{max}=150$, and sampling rates in $0.5R_t^{\min} \leq R_1+R_2 \leq 1.5R_t^{\max}$, i.e., $25 \leq R_1+R_2 \leq 225$. Note that the first tensor has 4 times as many elements as the second tensor (Table~V), so it can be expected that its rate will be (roughly) 4 times as large, i.e., $R_1=4\cdot R_2$. Plugging this back into the above inequality, we obtain that the ranges where $R_1$ and $R_2$ should be sampled are $20 \leq R_1 \leq 180$ and $5 \leq R_2 \leq 45$. For the proposed bit allocation method, we fitted the surface model~(\ref{eq:dist_total_form}) and then used~(\ref{eq:solution_SOP}) to allocate bits. 


Fig.~\ref{fig:distortion_surface_rates} shows the fitted distortion-rate surface from Fig.~\ref{fig:surface_class}, along with the intersections of this surface with three rate constraint planes $R_1+R_2=R_t$. The red dashed curve corresponds to $R_t=50$ Kbits, the blue dashed curve corresponds to $R_t=100$ Kbits and the yellow dashed curve corresponds to $R_t=150$ Kibts. Notice that the red curve is in a highly sloped part of the surface, blue curve in the medium-slope part, and yellow curve in a relatively flat part. The difference between optimal and sub-optimal bit allocations are most obvious in the highly sloped part (lowest $R_t$), since the impact on distortion is higher here than in the flatter parts of the surface. As we move towards the flatter parts of the surface, the difference between optimal and sub-optimal allocation decreases, as the impact on the resulting distortion becomes smaller and smaller. Our selection of $R_t$'s was meant to capture and illustrate these phenomena.

The rates allocated to each  tensor and the corresponding Top-1 accuracy for three different rate constraints are shown in Table~\ref{tbl:m4_rate_acc}. It is seen in the table that, as the rate constraint increases from $R_t=50$ Kbits to $R_t=150$ Kbits, the Top-1 accuracy increases for all allocation methods, as expected -- the less quantization noise, the better the accuracy. 

\begin{figure}[t]	
	\centering
	\centerline{\includegraphics[width=\columnwidth]{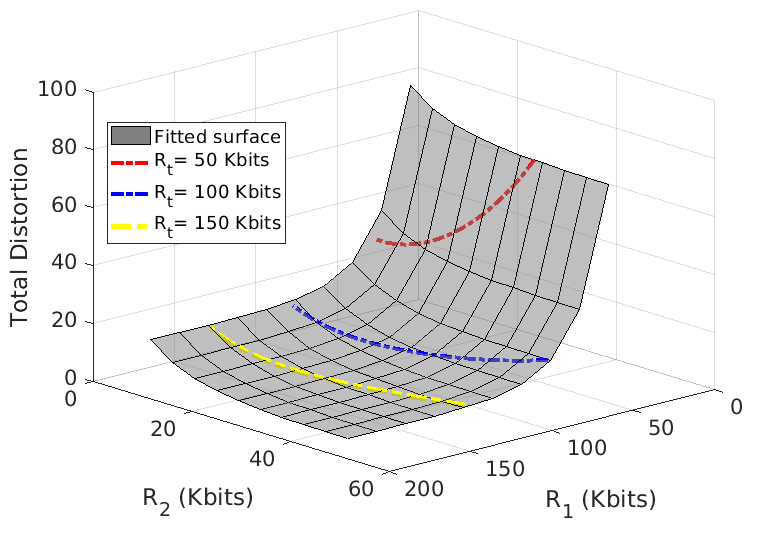}}
	\caption{
	Distortion-rate surface from Fig. 2. The dashed lines are the intersections of the distortion surface with planes $R_1+R_2= R_t$ for $R_t \in \{50,100,150\}$ Kbits.} 
	\label{fig:distortion_surface_rates}
\end{figure}

In all cases in Table~\ref{tbl:m4_rate_acc}, Method 1 (equal bit allocation) achieves the lowest Top-1 accuracy, while our proposed method achieves the highest accuracy, indicated in bold. Methods 2 and 3 are close behind, and for $R_t=100$ Kbits, Method 2 achieves the same Top-1 accuracy as our method, because it happens to allocate very similar rates to the two tensors as our method. 

\begin{table}[tbh]
\centering
\caption{The average rate (Kbits/tensor) for each tensor and the corresponding Top-1 accuracy. 
The Top-1 accuracy of the model without feature compression is 87.750\%. 
}
\begin{tabular}{|c|c|c|c|}
\hline
Method & $R_1$ & $R_2$ & Top-1 accuracy (\%) \\ \hline \hline
\multicolumn{4}{|c|}{$R_t=50$ Kbits}    \\ \hline
Ours     & 39.633 & 10.367  & \textbf{67.375
}        \\ \hline
Method 1 & 25.000 & 25.000 &   49.125     \\ \hline
Method 2 & 40.000 & 10.000 &   67.125         \\ \hline
Method 3 & 35.091 & 14.909 &   64.875      \\ \hline  \hline
\multicolumn{4}{|c|}{$R_t=100$ Kbits}   \\ \hline
Ours      & 78.272 & 21.728 & \textbf{84.875}        \\ \hline
Method 1  & 50.000 & 50.000 &   82.750     \\ \hline
Method 2  & 80.000 & 20.000 & \textbf{84.875}        \\ \hline
Method 3  & 70.182 & 29.818 & 84.375       \\ \hline \hline
\multicolumn{4}{|c|}{$R_t=150$ Kbits}   \\ \hline 
Ours      & 122.020 & 27.980 & \textbf{87.500}        \\ \hline
Method 1  & 75.000  & 75.000 & 85.750         \\ \hline
Method 2  & 120.000 & 30.000 &    87.250  \\ \hline
Method 3  & 105.273 & 44.727 &    86.750     \\ \hline
\end{tabular}
\label{tbl:m4_rate_acc}
\end{table}

\subsection{Multi-stream multi-task experiments}
\label{exp:Multi-task}
For multi-stream multi-task experiments we used the multi-task model from~\cite{our_icassp}, shown in Fig.~\ref{fig:curve}(c), which was trained on the Cityscapes dataset~\cite{cityscapes}. Cityscapes includes 2,975 training images with their corresponding semantic segmentation and disparity maps.  The multi-task model was trained to perform three tasks: semantic segmentation (Task 1), disparity estimation (Task 2) and input reconstruction (Task 3). The task-specific accuracy metrics are:  mean Intersection over Union (mIoU)~\cite{metrics} for semantic segmentation, Root Mean Squared Error (RMSE) in pixels~\cite{metrics} for disparity estimation, Peak Signal to Noise Ratio (PSNR) in dB for input reconstruction.  Since the annotations for the Cityscapes test set are not publicly available, the 500  images in the validation set are used as the test set, as in ~\cite{cambridge,NeurIPS_Sener}. From this set, 20\% of the data is randomly chosen 
to obtain distortion-rate surface parameters, and the remaining 80\%  is used for testing. 

The backbone (part of the model between the input and the stack block in Fig.~\ref{fig:curve}(c)) is similar to the backbone of YOLOv3~\cite{YOLOv3} and has 74 convolutional layers. Its weights were initialized using YOLOv3 weights. Models 1, 2, and 3 on the cloud side are based on the FC8 model~\cite{FC} and use the stacked deep features as inputs.
Cross-entropy loss~\cite{metrics} was used for semantic segmentation, and Mean Square Error (MSE) was used as the loss function for the other two tasks.
Following~\cite{cambridge}, the weighted sum of the mentioned losses is used as the total loss, and the weight of each loss was a trainable parameter.
Adam optimizer with the initial learning rate of 0.1 and rate decay by a factor of 0.85 every 20 epochs was used to train the model, end to end, for 250 epochs. 

\subsubsection{Bit allocation}
\label{subsubsubsec:2stream_bit_exp}
Two features tensors are taken from the model: one from layer 36 and the other from layer 61, as indicated by the vertical dashed curve in Fig.~\ref{fig:curve}(c). 
The dimension of the tensors are given in Table~\ref{tbl:feature_dimension_multi}. These features skip to the stack block and directly affect the output of the model. 
For a given task weight vector $\mathbf{w}=(w_1,w_2,w_3)$, we used scalarize-first approach (Section~\ref{sec:mop}) to scalarize the three task distortions into the total distortion $D_t$, and then fit the surface model~(\ref{eq:dist_total_form}) to $D_t$. Then, for a given rate constraint $R_t$, the bit allocation can be obtained from~(\ref{eq:solution_SOP}). 

\begin{table}[t]
\centering
\caption{The average rate (Kbits/tensor) for each tensor and the total distortion $D_t$, with $w_i=\frac{1}{3}$, achieved by various bit allocation methods for three rate constraints $R_t$.}
\vspace{0.2cm}
\label{tbl:benchmarks_dist_comparison}
\begin{tabular}{|c|c|c|c|}
\hline
 Method & $R_1$ & $R_2$ & Total distortion $D_t$ \\ \hline \hline  
 \multicolumn{4}{|c|}{$R_t=1000$ Kbits}    \\ \hline
Method 1 &    500.00  &  500.00 &   17.46 \\ \hline
Method 2 &    666.67  &   333.33 &  16.03       \\ \hline
Method 3 &   801.16  &   198.84  &  16.10      \\ \hline 
Ours     &   723.28   & 276.72    &\textbf{15.75}        \\ \hline \hline
\multicolumn{4}{|c|}{$R_t=1500$ Kbits}    \\ \hline
Method 1 &  750.00  & 750.00   &    14.99   \\ \hline
Method 2 &  1000.00 & 500.00   &    13.36          \\ \hline
Method 3 &  1201.75 & 298.25   &    11.88         \\ \hline
Ours     &  1215.44 & 284.56   &    \textbf{11.81}        \\ \hline \hline
\multicolumn{4}{|c|}{$R_t=2000$ Kbits}    \\ \hline
Method 1 &  1000.00    & 1000.00     &  13.21     \\ \hline
Method 2 &  1333.33 &  666.67  &    10.54         \\ \hline
Method 3 &  1602.33 & 397.67   &    8.89    \\ \hline
Ours    &   1701.67 & 298.33 &    \textbf{8.43}        \\ \hline
\end{tabular}
\end{table}

Table~\ref{tbl:benchmarks_dist_comparison} shows the average allocated bits and the corresponding total distortion $D_t$ with equal weights ($w_i=\frac{1}{3}$)  
for three rate constraints: $R_t\in\{1000, 1500, 2000\}$ Kbits. 
According to Table~\ref{tbl:feature_dimension_multi}
, the tensor sizes here ($32\times64\times256$ and $16\times32\times512$) are larger compared to the tensors from the DenseNet model in Section~\ref{exp:single-task} ($28\times28\times128$ and $28\times28\times32$), so the tested rates are larger as well. The results for Methods 1-3 are also shown in the table.  The lowest total distortion under each $R_t$ is indicated in bold. In all cases, our bit allocation achieves the lowest distortion. At $R_t=1500$ Kbits, Method 3 happens to produce very similar rates to our method, and the resulting distortions are correspondingly similar.  Nonetheless, the proposed solution gives the lowest distortion.  

Since $R_t=1000$ Kbits is the lowest among the tested rates here, the total distortion values at $R_t=1000$ Kbits are higher than at the other two rates, as seen in Table~\ref{tbl:benchmarks_dist_comparison}. In fact, near $R_t=1000$ Kbits, the total distortion surface $D_t$ is more sloped than near the other two rates (similar to what is shown in Fig.~\ref{fig:distortion_surface_rates}), and the difference between optimal and sub-optimal bit allocations will be more obvious. Therefore, 
$R_t=1000$ Kbits is chosen for the remaining experiments.


Next we examine the effects of task weights $w_i$. When the weights change, the total distortion $D_t$ and its approximating surface 
will change, so the proposed method will find different rates in~(\ref{eq:solution_SOP}). Meanwhile, the three benchmarks keep their bit allocations unchanged, because the number of elements in the tensors and their variance stay the same. The task-specific accuracies also stay the same for the benchmarks, but their total distortion changes according to the new weights.
Table~\ref{tbl:benchmarks_result_comparison} has two parts: in the top part, task-specific accuracies and the allocated rates ($R_1$ and $R_2$) are shown for 
uncompressed features, benchmarks and five sets of weights, while the bottom part shows total distortion $D_t$ for two sets of weights, for a total rate constraint of $R_t=1000$ Kbits. The best task-specific accuracies (
excluding the accuracies without compression) in the top part, and the lowest total distortions in the bottom part, are indicated in bold.

The first set of weights are the equal weights $(w_1,w_2,w_3)=(\frac{1}{3},\frac{1}{3},\frac{1}{3})$. The second set of weights is chosen inversely proportional to the average task distortion over the 100 measured rate tuples. Specifically, 
if $\overline{D_i}$ is the average distortion of the $i$-th task over the sampled 100 rate tuples, the corresponding task weight  
is obtained as:
\vspace{-5pt}
\begin{equation}
    w_i = \frac{\frac{1}{\overline{D_i}}}{\sum_{j=1}^{k} \frac{1}{\overline{D_j}}}
\end{equation}
This way, the weighted distortions in~(\ref{eq:scalarization}) are equalized. For our experiment, these weights were computed as $(w_1,w_2,w_3)=(0.5,0.48,0.02)$. The last three columns in the top part of the table show the test accuracies obtained when bits are allocated to optimize only one task at a time. These are included to support the discussion of the Pareto set. 



\begin{table*}[tb!]
\centering
\caption{Rates allocated to each tensor, task-specific accuracies and total distortion  for $R_t$=1000 Kbits. Higher numbers ($\uparrow$) are better for mIoU and PSNR; lower numbers ($\downarrow$) are better for RMSE.}
\vspace{0.2cm}
\label{tbl:benchmarks_result_comparison}
\begin{tabular}{|c|c|c|c|c|c|c|c|c|c|}
\hline
$R_t=1000$  & \begin{tabular}[c]{@{}c@{}} No \\compression \\ \end{tabular} & \begin{tabular}[c]{@{}c@{}}Method\\ 1\end{tabular} & \begin{tabular}[c]{@{}c@{}}Method \\ 2\end{tabular} & \begin{tabular}[c]{@{}c@{}}Method \\ 3\end{tabular} & \begin{tabular}[c]{@{}c@{}}Ours\\ ($\frac{1}{3}, \frac{1}{3}, \frac{1}{3}$)\end{tabular} & \begin{tabular}[c]{@{}c@{}}Ours\\ (0.5, 0.48, 0.02)\end{tabular} & \begin{tabular}[c]{@{}c@{}}Ours\\ (1, 0, 0)\end{tabular} & \begin{tabular}[c]{@{}c@{}}Ours\\ (0, 1, 0)\end{tabular} & \begin{tabular}[c]{@{}c@{}}Ours\\ (0, 0, 1)\end{tabular}  \\ \hline \hline 
Task 1 - mIoU (\%) $\enspace \uparrow$ & 63.02 & 61.91 & 62.16 & 61.75 & 62.15& 62.20 & \textbf{62.22} & 62.09 & 44.23  \\ \hline
Task 2 - RMSE (px)  $\downarrow$& 7.80 & \textbf{7.85} & 7.87 & 8.07 & 7.91& \textbf{7.85} & 7.86 & \textbf{7.85} & 16.43   \\ \hline
Task 3 - PSNR (dB) $\uparrow$ & 39.97 & 19.98 & 21.61 & 22.79 & 22.16 & 20.87& 21.06& 20.34 & \textbf{24.20}   \\ \hline \hline 
$R_1$ & - & 500.00 & 666.67 & 801.16 & 723.28 & 601.23 & 616.72 & 548.70  & 986.11     \\ \hline
$R_2$ & - & 500.00 & 333.33 & 198.84 & 276.72 & 398.77 & 383.28 & 451.30 &  13.89  \\ \hline
\multicolumn{9}{c}{\vspace{0.1cm}}\\
\end{tabular} 

\begin{tabular}{|c|c|c|c|c|}
\hline
$R_t=1000$  & \begin{tabular}[c]{@{}c@{}}Method 1\end{tabular} & \begin{tabular}[c]{@{}c@{}}Method 2\end{tabular} & \begin{tabular}[c]{@{}c@{}}Method 3\end{tabular} &  \begin{tabular}[c]{@{}c@{}}Ours\end{tabular}   \\ \hline \hline \\[-1em]
\begin{tabular}[c]{@{}c@{}}
$D_t$ with $(\frac{1}{3}, \frac{1}{3}, \frac{1}{3})$\end{tabular} & 17.46 & 16.03 & 16.10 & \textbf{15.75}   \\ [-1em]
\\
\hline
\begin{tabular}[c]{@{}c@{}}
\\[-1em]
$D_t$ with $(0.5, 0.48, 0.02)$\end{tabular} & 2.20 & 2.01 & 3.48 & \textbf{1.91} \\ [-1em]
\\ \hline
\end{tabular}%
\end{table*}

Based on the rates in the last three columns of the top part of Table~\ref{tbl:benchmarks_result_comparison} and Theorem~\ref{th:2stream_pareto}, we conclude that the Pareto set for this problem is located on the line $R_1+R_2=R_t=1000$ Kbits, between the points $(R_1^{\min},R_2^{\max})=(548.70, 451.30)$ and $(R_1^{\max},R_2^{\min})=(986.11, 13.89)$. Hence, the rate allocations produced by our method for all five sets of weights, as well as the rates produced by Methods 2 and 3, are in the Pareto set. The comparison between any two among these seven rate allocations shows that some task accuracies improve, while others degrade. That is the nature of the Pareto set -- one cannot achieve improvement in some tasks without degrading others. However, we note that Method 1 produces rates $(R_1,R_2)=(500,500)$, which are outside of the Pareto set. Hence, it should be possible to improve some task accuracies compared to this allocation, without hurting others. Indeed, we find that our method with $(w_1,w_2,w_3)=(0.5,0.48,0.02)$ and $(w_1,w_2,w_3)=(0,1,0)$ finds allocations that improve the accuracies of Tasks 1 and 3 without hurting the accuracy of Task 2, compared to Method 1. 

Among the points on the Pareto set (Methods 2 and 3, and our allocations) it is impossible to say one solutions is uniformly better than the others, because of the nature of Pareto-optimal solutions -- some accuracies improve, others degrade. One sensible way to compare the solutions is in terms of the total (scalarized) distortion $D_t$ from~(\ref{eq:scalarization}), which is shown in the bottom part of Table~\ref{tbl:benchmarks_result_comparison}. Based on this these results, our allocations produce the lowest $D_t$, as expected. 

Comparing task-specific accuracies with and without compression  shows that that semantic segmentation and disparity map estimation are not very sensitive to feature compression, as the drop in their accuracies due to feature compression is relatively small. However, input reconstruction seems to be significantly impacted by feature compression in our multi-task model, since the drop in the accuracy of input reconstruction due to feature compression is comparatively large.

\subsubsection{Pareto set for a $2 \times 2$ system}
\label{subsubsubsec:2stream_pareto_exp}
In order to gain further insight into the nature of the Pareto set, we focus now on a $2\times2$ CI system, which is the same as the system in the previous subsection, but with Task 3 (input reconstruction) removed. So, we are left with two tensors and two tasks. This allows us to plot the Pareto set in the $R_1$-$R_2$ rate plan, and also plot the Pareto front in the $D_1$-$D_2$ distortion plane. The rate constraint is again $R_t=1000$ Kbits. Based on Table~\ref{tbl:benchmarks_result_comparison} and Theorem~\ref{th:2stream_pareto}, the Pareto set is now located on the line $R_1+R_2=R_t=1000$, between the points $(R_1^{\min},R_2^{\max})=(548.70, 451.30)$ and $(R_1^{\max},R_2^{\min})=(616.72, 383.28)$.

The Pareto set is shown in Fig.~\ref{fig:m4_2rates_results}(a). It is located on the rate constraint line between the two points specified above, which are displayed as blue and cyan diamonds. These two points are obtained by minimizing individual task distortions. Also shown are 1000 green points on the Pareto set. These are obtained by numerically  solving~(\ref{eq:scalarization}), using Matlab's \texttt{fmincon} function for 1000 weight pairs. These discrete samples are what we can typically obtain from various numerical Pareto solvers. But Theorem~\ref{th:2stream_pareto} allows us to characterize the entire Pareto set analytically.

The Pareto front is the set of distortions obtained as the rates sweep the Pareto set. It is shown in the $D_1$-$D_2$ distortion plane 
in Fig.~\ref{fig:m4_2rates_results}(b). In this figure, the gray points are the distortion points obtained with $R_1+R_2<R_t$. The red and green points are the distortion points obtained with $R_1+R_t=R_t$. This intuitively shows why the Pareto-optimal solutions for our problem satisfy $R_1+R_2=R_t$, and is a consequence of the monotonicity of the distortion surfaces. The green points represent the Pareto front, and are obtained by computing distortions for the 1000 numerically sampled Pareto-optimal solutions in Fig.~\ref{fig:m4_2rates_results}(a). The diamonds in Fig.~\ref{fig:m4_2rates_results}(b)  correspond to the diamonds in  Fig.~\ref{fig:m4_2rates_results}(a), and represent distortion minima of the two tasks. 
For every point on the Pareto front, 
reducing one task distortion increases  the other distortion, and this can clearly be seen in Fig.~\ref{fig:m4_2rates_results}(b). 


\begin{figure}[t] 
\begin{tabular}{ccc}
\begin{subfigure}{\columnwidth}\centering\includegraphics[width=7.5cm]{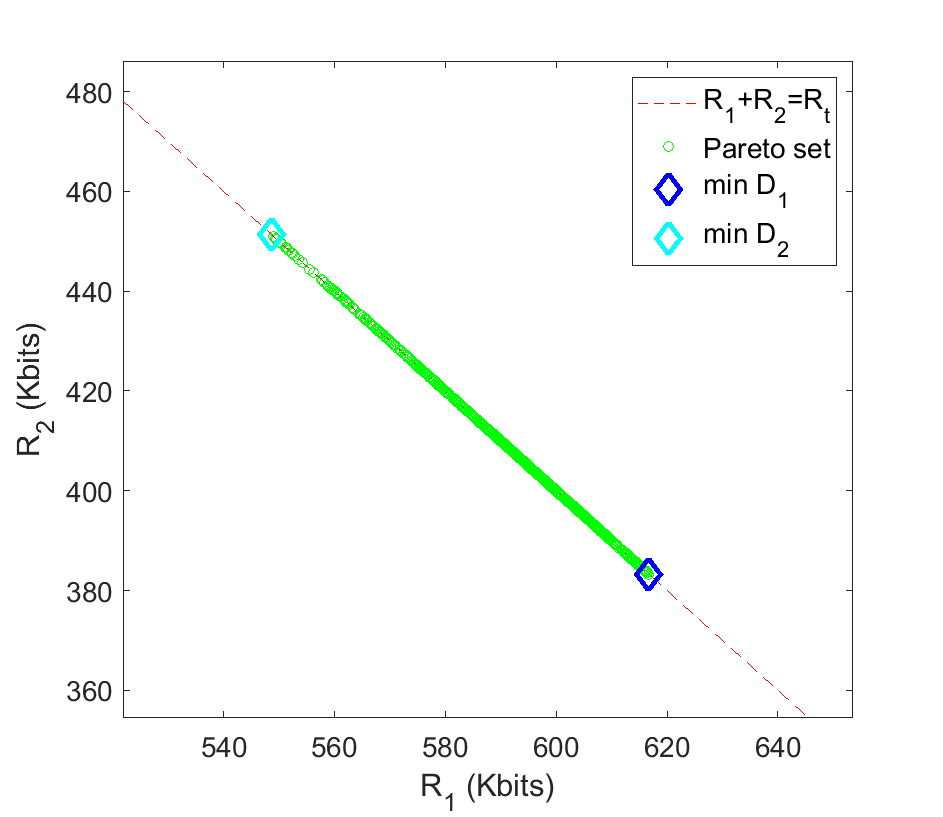}\caption{}\label{fig:2rates_pareto_set}\end{subfigure} \\
\begin{subfigure}{\columnwidth}\centering\includegraphics[width=7.5cm]{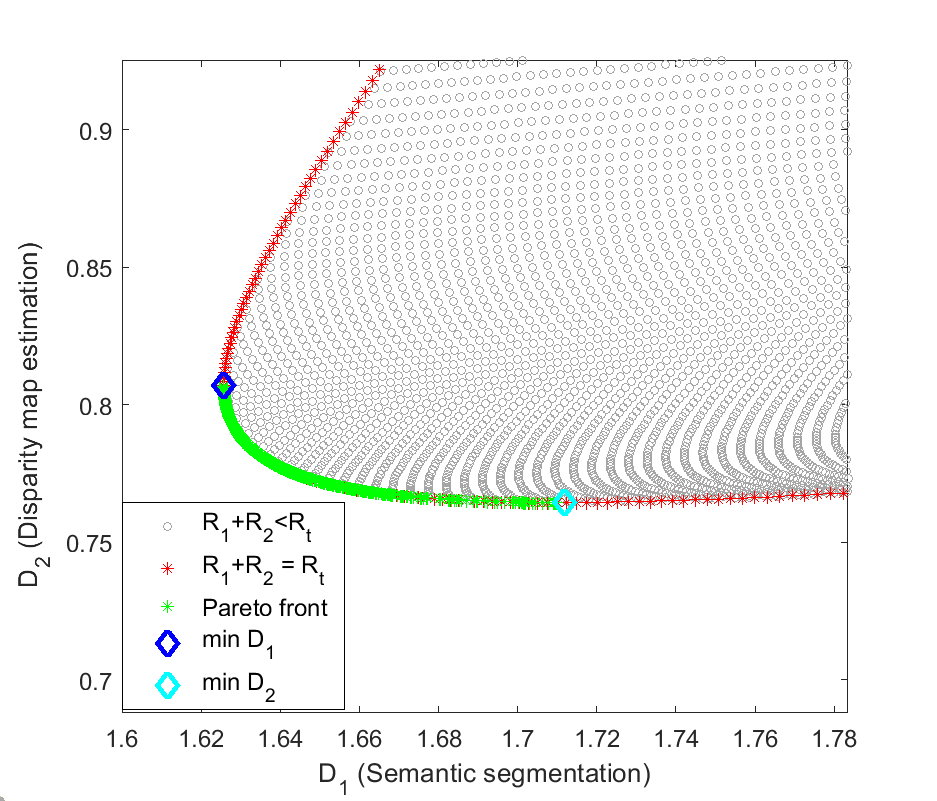}\caption{}\label{fig:2rates_pareto_front}\end{subfigure} 
\end{tabular}
\caption{$2 \times 2$ CI system: (a) Pareto set (green) in the rate plane; (b) Pareto front (green) in the distortion plane. 
}
\label{fig:m4_2rates_results}
\end{figure}

\subsubsection{Pareto set for a $3 \times 2$ system}
\label{subsubsubsec:3stream_pareto_exp}
Here we again use the model in Fig.~\ref{fig:multi-task_model}, but the split point is moved to just before the stack layer (output of layer 74) where there are three tensors to be transferred. We keep Tasks 1 and 2 from the previous section, thereby creating a $3 \times 2$ CI system. Following the procedure described before Theorem~\ref{th:3stream_pareto}, we find rate extrema for the total rate constraint of $R_t=1000$~Kbits. Theorem~\ref{th:3stream_pareto} then states that the Pareto set is bounded by a polygonal region obtained as the intersection of the rate constraint plane and rate extrema cube. This region is shown in Fig.~\ref{fig:m4_3rates_results} as orange points. The green points show 1000 samples from the Pareto set obtained numerically, as in the previous section. As shown in the figure, the Pareto set is fully contained in the orange region, as predicted by Theorem~\ref{th:3stream_pareto}. On the other hand, this is not the tightest possible bound on the Pareto set; tightening this bound is a topic for future research. 

\begin{figure}[t] 
\centering
\includegraphics[width=\columnwidth]{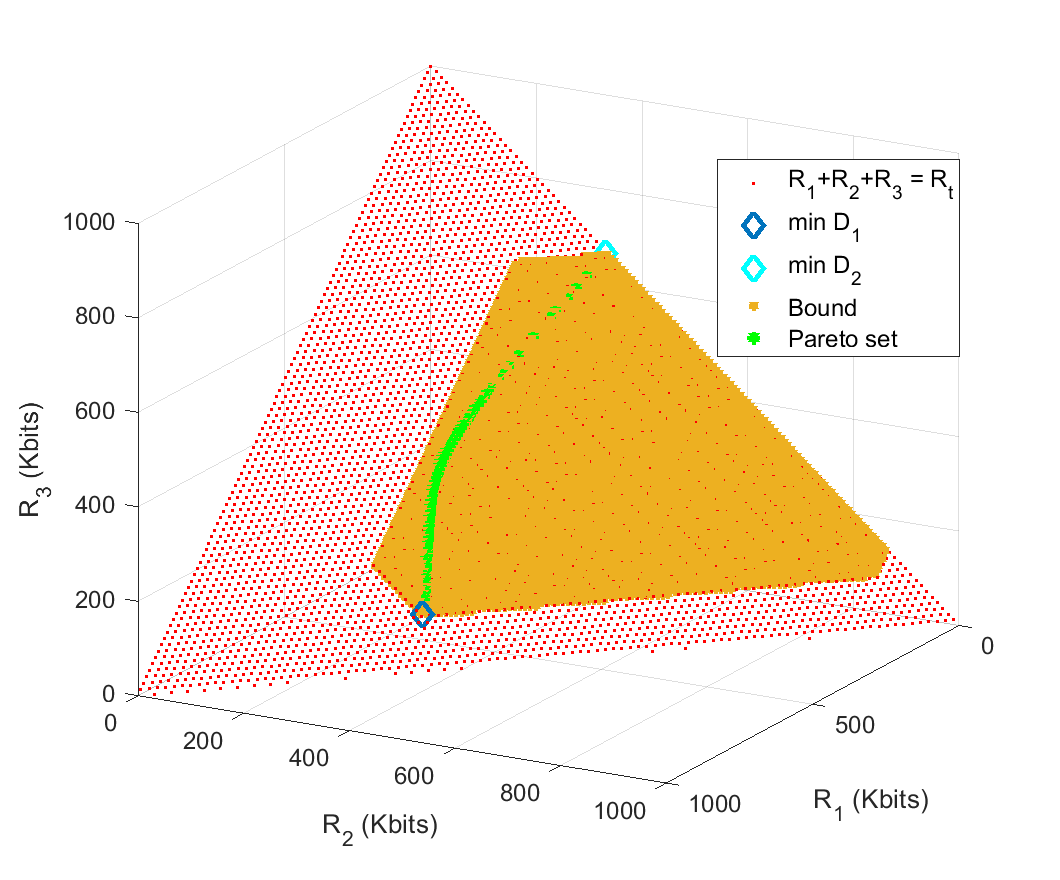}
\caption{$3 \times 2$ CI system: red dots represent the $R_1+R_2+R_3=R_t=1000$ Kbits plane;  orange hexagon is the bound on the Pareto set from Theorem~\ref{th:3stream_pareto}; 
green dots are the numerically-obtained samples from the Pareto set; blue and cyan diamonds correspond to the minima of two task distortions.  
}
\label{fig:m4_3rates_results}
\end{figure}

\begin{table}[h]
\centering
\caption{
Dimensions of the input image and deep feature tensors for the single-task (DenseNet) model.}
\label{tbl:feature_dimension_densenet}
\begin{tabular}{|c|c|}
\hline
            & Height$\times$Width$\times$Channels \\ \hline \hline
Input image & 224 $\times$ 224 $\times$ 3                  \\ \hline
Tensor 1 (layer 15)    & 28 $\times$ 28 $\times$128                      \\ \hline
Tensor 2 (layer 17)    & 28$\times$28$\times$32                       \\ \hline
\end{tabular}
\end{table}

\begin{table}[h]
\centering
\caption{
Dimensions of the input image and deep feature tensors for the  multi-task model.}
\label{tbl:feature_dimension_multi}
\begin{tabular}{|c|c|}
\hline
            & Height$\times$Width$\times$Channels \\ \hline \hline
Input image & 256 $\times$ 512 $\times$ 3                  \\ \hline
Tensor 1 (layer 36)    & 32 $\times$ 64 $\times$ 256                      \\ \hline
Tensor 2 (layer 61)    & 16 $\times$ 32 $\times$ 512                       \\ \hline
\end{tabular}
\end{table}

\section{Conclusion}
\label{sec:conclusion}
In this paper we studied the bit allocation problem for multi-stream multi-task CI systems. 
A convex approximation to the distortion-rate surface was proposed, which led to the closed-form solution for bit allocation in single-task systems and scalarized multi-task systems. In addition, analytical characterization of the Pareto set was obtained for 2-stream multi-task systems, and a bound on the Pareto set was derived for 3-stream 2-task systems.  

The analysis was supplemented by experiments on several CI systems involving tasks such as image classification, semantic segmentation, disparity estimation, and input reconstruction. The obtained bit allocation solutions were compared against several alternatives. The results demonstrated the advantages of the presented solutions compared to the alternatives. In addition, Pareto set was examined in several cases to obtain further insight into the nature of the bit allocation problem for CI systems. 
\begin{appendices}
\section{Proof of Theorem \ref{th:2stream_pareto}}
\label{appendix:2stream_proof}
\textbf{Theorem.} For a 2-stream $k$-task CI  system, any point on the line $R_1+R_2=R_t$ between $(R_1^{\min},R_2^{\max})$ and $(R_1^{\max},R_2^{\min})$ is Pareto-optimal. Moreover, there are no Pareto-optimal solutions outside of this line segment.
\begin{proof}
Let us parametrize the line $R_1 + R_2 = R_t$ as $R_1 = r$, $R_2 = R_t - r$. This way, the intersection of each distortion surface $D_i(\mathbf{R})$ and the plane defined by $R_1 + R_2 = R_t$ becomes a function of one variable, $D_i(r)$. From~(\ref{eq:dist_total_form}), distortion surfaces are strictly convex, and so are distortion curves $D_i(r)$. We will use the fact that a strictly convex function keeps increasing as we move further away from its minimum, on either side.  Let $r_{\min}=R_1^{\min}$ and $r_{\max}=R_1^{\max}$. According to the setup of the theorem, these are the points at which some distortion curves reach their minimum. Let $r_{\min}$ be the minimum of $D_l(r)$ and $r_{\max}$ be the minimum of $D_u(r)$. 

First, we show that every point between $r_{\min}$ and $r_{\max}$ is Pareto-optimal. Let $r'$ be any point between $r_{\min}$ and $r_{\max}$.   
Assume that $r'$ is \emph{not} a Pareto-optimal solution. In this case, there must exist another point $r''$ such that 
\begin{equation}
\label{eq:counter_condition}
    \forall i, \, D_i(r'') \leq D_i(r') \quad \text{and} \quad \exists j, \, D_j(r'') < D_j(r').
\end{equation}
Let $r''$ be such a point. There are two possibilities: either $r'' < r'$ or $r'' > r'$. If the former is the case, then, because $r_{\min} \leq r' \leq r_{\max}$, it must be that $r'' < r' \leq r_{\max}$. Then, by the strict convexity of $D_u(r)$ and the fact that $r_{\max}$ is the minimum of $D_u(r)$, we have
\begin{equation}
\label{eq:upper_contradiction}
    D_u(r'') > D_u(r') \geq D_u (r_{\max}).
\end{equation}
On the other hand, if $r'' > r'$, then $r_{\min}<r'<r''$. In this case, by the strict convexity of $D_l(r)$ and the fact that $r_{\min}$ is the minimum of $D_l(r)$, we have 
\begin{equation}
\label{eq:lower_contradiction}
    D_l(r'') > D_l(r') \geq D_l (r_{\min}).
\end{equation}
From~(\ref{eq:upper_contradiction}) and~(\ref{eq:lower_contradiction}) we conclude that there is at least one $i$ such that $D_i(r'')>D_i(r')$, thereby contradicting~(\ref{eq:counter_condition}). Hence, $r'$ is Pareto-optimal.

To complete the proof, we show that no point outside $[r_{\min},r_{\max}]$ can be Pareto-optimal. Recall that by the definition of $r_{\min}$ and $r_{\max}$ (more specifically, by the definition of $R_1^{\min}$ and $R_1^{\max}$ in Section~\ref{sec:Pareto_2xk}), the minimum of each distortion curve $D_i(r)$ lies in $[r_{\min},r_{\max}]$. Let $r'$ be any given point outside $[r_{\min},r_{\max}]$ and assume that it \emph{is} Pareto-optimal. In that case, there must \emph{not} exist another point $r''$ such that~(\ref{eq:counter_condition}) is true. 

Since $r' \notin [r_{\min},r_{\max}]$, then either $r'<r_{\min}$ or $r'>r_{\max}$. If the former is the case, then $r'$ is ``to the left'' of all distortion minima. Hence, by the strict convexity of $D_i(r)$, we have
\begin{equation}
\label{eq:left_contradiction}
    \forall i, \, D_i(r') > D_i(r_{\min}) \geq \min_r D_i (r),
\end{equation}
so $r''=r_{\min}$ is the point that satisfies~(\ref{eq:counter_condition}). On the other hand, if $r'>r_{\max}$, then $r'$ is ``to the right'' of all distortion minima, so by the strict convexity of $D_i(r)$ we have
\begin{equation}
\label{eq:right_contradiction}
    \forall i, \, \min_r D_i (r) \leq D_i(r_{\max}) < D_i(r'),
\end{equation}
so $r''=r_{\max}$ is the point that satisfies~(\ref{eq:counter_condition}). Either way, we can find a point that satisfies~(\ref{eq:counter_condition}), so $r'$ cannot be Pareto-optimal. This completes the proof.

\end{proof}



\section{Extrema of $R_1, R_2$, and $R_3$ for a $3 \times 2$ CI system}
\label{appendix:extreme_3x2}

\begin{equation}
\small
\label{eq:r1_ext}
\begin{gathered}
R_1^{ext1} = \frac{R_t\beta_{2,3} \beta_{2,2} - \beta_{2,2} \log_{2}(\alpha_{2,3} \beta_{2,3}) - \beta_{2,3} \log_{2}(\alpha_{2,2} \beta_{2,2})}{ \beta_{2,1} \beta_{2,2} + \beta_{2,1} \beta_{2,3} + \beta_{2,3} \beta_{2,2} } +\\
+ \frac{\beta_{2,2} \log_{2}(\alpha_{2,1} \beta_{2,1}) +  \beta_{2,3} \log_{2}(\alpha_{2,1} \beta_{2,1})}{ \beta_{2,1} \beta_{2,2} + \beta_{2,1} \beta_{2,3} + \beta_{2,3} \beta_{2,2} }\\
R_1^{ext2} = \frac{R_t\beta_{2,3} \beta_{1,2} - \beta_{1,2} \log_{2}(\alpha_{2,3} \beta_{2,3}) - \beta_{2,3} \log_{2}(\alpha_{1,2} \beta_{1,2})}{ \beta_{1,2} \beta_{2,1} + \beta_{1,1} \beta_{2,3} + \beta_{2,3} \beta_{1,2} } +\\
\frac{  \beta_{1,2} \log_{2}(\alpha_{2,1} \beta_{2,1}) +  \beta_{2,3} \log_{2}(\alpha_{1,1} \beta_{1,1})}{ \beta_{1,2} \beta_{2,1} + \beta_{1,1} \beta_{2,3} + \beta_{2,3} \beta_{1,2} } \\
R_1^{ext3} = \frac{R_t\beta_{1,3} \beta_{2,2} - \beta_{2,2} \log_{2}(\alpha_{1,3} \beta_{1,3}) - \beta_{1,3} \log_{2}(\alpha_{2,2} \beta_{2,2}) }{ \beta_{1,1} \beta_{2,2} + \beta_{2,1} \beta_{1,3} + \beta_{1,3} \beta_{2,2} } +\\
\frac{  \beta_{2,2} \log_{2}(\alpha_{1,1} \beta_{1,1}) +  \beta_{1,3} \log_{2}(\alpha_{2,1} \beta_{2,1})}{ \beta_{1,1} \beta_{2,2} + \beta_{2,1} \beta_{1,3} + \beta_{1,3} \beta_{2,2} }\\
R_1^{ext4} = \frac{R_t\beta_{1,3} \beta_{1,2} - \beta_{1,2} \log_{2}(\alpha_{1,3} \beta_{1,3}) - \beta_{1,3} \log_{2}(\alpha_{1,2} \beta_{1,2})}{ \beta_{1,1} \beta_{1,2} + \beta_{1,1} \beta_{1,3} + \beta_{1,3} \beta_{1,2} } +\\
\frac{\beta_{1,2} \log_{2}(\alpha_{1,1} \beta_{1,1}) +  \beta_{1,3} \log_{2}(\alpha_{1,1} \beta_{1,1})}{ \beta_{1,1} \beta_{1,2} + \beta_{1,1} \beta_{1,3} + \beta_{1,3} \beta_{1,2} } \\
R_1^{\min} = \min\big\{ R_1^{ext1}, R_1^{ext2}, R_1^{ext3}, R_1^{ext4}, R_t \big\} \\
R_1^{\max} = \max\big\{ R_1^{ext1}, R_1^{ext2}, R_1^{ext3}, R_1^{ext4}, 0 \big\}
\end{gathered}
\end{equation}

\small
\begin{equation}
\label{eq:r2_ext}
\begin{gathered}
R_2^{ext1} = \frac{R_t\beta_{2,1} \beta_{2,3} - \beta_{2,3} \log_{2}(\alpha_{2,1} \beta_{2,1}) - \beta_{2,1} \log_{2}(\alpha_{2,3} \beta_{2,3})}{ \beta_{2,2} \beta_{2,3} + \beta_{2,2} \beta_{2,1} + \beta_{2,1} \beta_{2,3} } +\\
\frac{ \beta_{2,1} \log_{2}(\alpha_{2,2} \beta_{2,2}) +  \beta_{2,3} \log_{2}(\alpha_{2,2} \beta_{2,2})}{ \beta_{2,2} \beta_{2,3} + \beta_{2,2} \beta_{2,1} + \beta_{2,1} \beta_{2,3} }\\
R_2^{ext2} = \frac{R_t\beta_{2,1} \beta_{1,3} - \beta_{1,3} \log_{2}(\alpha_{2,1} \beta_{2,1}) - \beta_{2,1} \log_{2}(\alpha_{1,3} \beta_{1,3})}{ \beta_{1,3} \beta_{2,2} + \beta_{1,2} \beta_{2,1} + \beta_{2,1} \beta_{1,3} } +\\
\frac{\beta_{1,3} \log_{2}(\alpha_{2,2} \beta_{2,2}) +  \beta_{2,1} \log_{2}(\alpha_{1,2} \beta_{1,2})}{ \beta_{1,3} \beta_{2,2} + \beta_{1,2} \beta_{2,1} + \beta_{2,1} \beta_{1,3} }\\
R_2^{ext3} = \frac{R_t\beta_{1,1} \beta_{2,3} - \beta_{2,3} \log_{2}(\alpha_{1,1} \beta_{1,1}) - \beta_{1,1} \log_{2}(\alpha_{2,3} \beta_{2,3})}{ \beta_{1,2} \beta_{2,3} + \beta_{1,1} \beta_{2,2} + \beta_{1,1} \beta_{2,3} } +\\
\frac{ \beta_{2,3} \log_{2}(\alpha_{12} \beta_{1,2}) +  \beta_{1,1} \log_{2}(\alpha_{2,2} \beta_{2,2})}{ \beta_{1,2} \beta_{2,3} + \beta_{1,1} \beta_{2,2} + \beta_{1,1} \beta_{2,3} } \\
R_2^{ext4} = \frac{R_t\beta_{1,1} \beta_{1,3} - \beta_{1,3} \log_{2}(\alpha_{1,1} \beta_{1,1}) - \beta_{1,1} \log_{2}(\alpha_{1,3} \beta_{1,3})}{ \beta_{1,2} \beta_{1,3} + \beta_{1,1} \beta_{1,2} + \beta_{1,1} \beta_{1,3} } +\\
\frac{\beta_{1,1} \log_{2}(\alpha_{1,2} \beta_{1,2}) +  \beta_{1,3} \log_{2}(\alpha_{1,2} \beta_{1,2})}{ \beta_{1,2} \beta_{1,3} + \beta_{1,1} \beta_{1,2} + \beta_{1,1} \beta_{1,3} }\\
R_2^{\min} = \min\big\{ R_2^{ext1}, R_2^{ext2}, R_2^{ext3}, R_2^{ext4}, R_t \big\} \\
R_2^{\max} = \max\big\{ R_2^{ext1}, R_2^{ext2}, R_2^{ext3}, R_2^{ext4}, 0 \big\}
\end{gathered}
\end{equation}

\begin{equation}
\label{eq:r3_ext}
\begin{gathered}
R_3^{ext1} = \frac{R_t\beta_{2,1} \beta_{2,2} - \beta_{2,2} \log_{2}(\alpha_{2,1} \beta_{2,1}) - \beta_{2,1} \log_{2}(\alpha_{2,2} \beta_{2,2}) }{ \beta_{2,3} \beta_{2,2} + \beta_{2,3} \beta_{2,1} + \beta_{2,1} \beta_{2,2} } +\\
\frac{\beta_{2,2} \log_{2}(\alpha_{2,3} \beta_{2,3}) +  \beta_{2,1} \log_{2}(\alpha_{2,3} \beta_{2,3})}{ \beta_{2,3} \beta_{2,2} + \beta_{2,3} \beta_{2,1} + \beta_{2,1} \beta_{2,2} } \\
R_3^{ext2} = \frac{R_t\beta_{2,1} \beta_{1,2} - \beta_{1,2} \log_{2}(\alpha_{2,1} \beta_{2,1}) - \beta_{2,1} \log_{2}(\alpha_{1,2} \beta_{1,2})}{ \beta_{2,3} \beta_{1,2} + \beta_{1,3} \beta_{2,1} + \beta_{2,1} \beta_{1,2} } +\\
\frac{  \beta_{1,2} \log_{2}(\alpha_{2,3} \beta_{2,3}) +  \beta_{2,1} \log_{2}(\alpha_{1,3} \beta_{1,3})}{ \beta_{2,3} \beta_{1,2} + \beta_{1,3} \beta_{2,1} + \beta_{2,1} \beta_{1,2} }\\
R_3^{ext3} = \frac{R_t\beta_{1,1} \beta_{2,2} - \beta_{2,2} \log_{2}(\alpha_{1,1} \beta_{1,1}) - \beta_{1,1} \log_{2}(\alpha_{2,2} \beta_{2,2})}{ \beta_{1,3} \beta_{2,2} + \beta_{2,3} \beta_{1,1} + \beta_{1,1} \beta_{2,2} } +\\
\frac{\beta_{2,2} \log_{2}(\alpha_{1,3} \beta_{1,3}) +  \beta_{1,1} \log_{2}(\alpha_{2,3} \beta_{2,3})}{ \beta_{1,3} \beta_{2,2} + \beta_{2,3} \beta_{1,1} + \beta_{1,1} \beta_{2,2} }\\
R_3^{ext4} = \frac{R_t\beta_{1,1} \beta_{1,2} - \beta_{1,2} \log_{2}(\alpha_{1,1} \beta_{1,1}) - \beta_{1,1} \log_{2}(\alpha_{1,2} \beta_{1,2}) }{ \beta_{1,3} \beta_{1,2} + \beta_{1,3} \beta_{1,1} + \beta_{1,1} \beta_{1,2}} +\\
\frac{\beta_{1,2} \log_{2}(\alpha_{1,3} \beta_{1,3}) +  \beta_{1,1} \log_{2}(\alpha_{1,3} \beta_{1,3})}{ \beta_{1,3} \beta_{1,2} + \beta_{1,3} \beta_{1,1} + \beta_{1,1} \beta_{1,2}}\\
R_3^{\min} = \min\big\{R_3^{ext1}, R_3^{ext2}, R_3^{ext3}, R_3^{ext4}, R_t \big\} \\ 
R_3^{\max} = \max\big\{ R_3^{ext1}, R_3^{ext2}, R_3^{ext3}, R_3^{ext4}, 0 \big\}
\end{gathered}
\end{equation}

\normalsize
\section{Proof of Theorem \ref{th:3stream_pareto}}
\label{appendix:3stream_proof}
\textbf{Theorem.} For a 3-stream 2-task CI system, there are no Pareto-optimal solutions outside of $\mathcal{C} \, \cap \, \mathcal{P}$ (Fig.~\ref{fig:cube}).

\begin{proof}
Let $\mathbf{R}'=(R'_1,R'_2)$ be any given point outside $\mathcal{C} \, \cap \, \mathcal{P}$. Assume that $\mathbf{R}'$ \emph{is}  Pareto-optimal for~(\ref{eq:rdo_mop}). If that is the case, we must have $\mathbf{R}' \in S$,  and by Theorem~\ref{th:kkt}, $\mathbf{R}'$ must satisfy KKT conditions~(\ref{eq:kkt}) for some  $\mathbf{w}=(w_1,w_2)$ with $w_1,w_2>0$. However, by construction (Section~\ref{sec:Pareto_3x2}), $\mathcal{C} \, \cap \, \mathcal{P}$ includes all feasible points $\mathbf{R} \in S$ that satisfy KKT conditions~(\ref{eq:mop_system}) with $\frac{w_1}{w_2} >0$, which includes all $w_1,w_2>0$. Hence, $\mathbf{R}'$ cannot be Pareto-optimal.  
\end{proof}

\end{appendices}

%




\ifCLASSOPTIONcaptionsoff
  \newpage
\fi



\bibliographystyle{IEEEtran}
\bibliography{IEEEabrv,refs}
%

%
\begin{IEEEbiography}[{\includegraphics[width=1in,height=1.25in,clip,keepaspectratio]{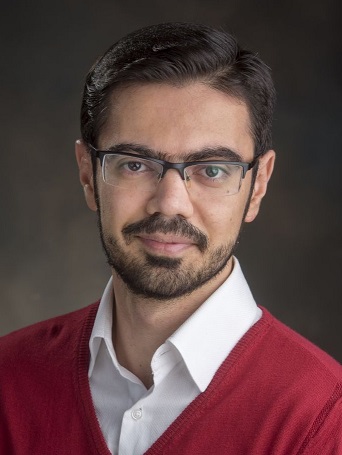}}]{Saeed Ranjbar Alvar}
(Student Member, IEEE) received the B.Sc degree in electrical engineering from University of Tabriz, Iran, in 2014 and the M.Sc. degree in electrical and electronics engineering from Middle East Technical University, Turkey, in 2016. He is currently a Ph.D. candidate at Simon Fraser University, Canada. He is the recipient of the best student paper award and the IEEE SPS travel grant at ICIP’19 in Taipei, Taiwan. His current research interests include Deep learning, image/video processing and computer vision.
\end{IEEEbiography}

\begin{IEEEbiography}[{\includegraphics[height=1.2in,keepaspectratio]{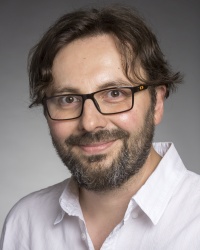}}]{Ivan V. Baji\'{c}} (S'99-M'04-SM'11) received the Ph.D. degree in electrical engineering from Rensselaer Polytechnic Institute, Troy, NY, in 2003. He is a Professor of Engineering Science and co-director of the Multimedia Lab at Simon Fraser University, Burnaby, BC, Canada. His research interests include signal processing and machine learning with applications to multimedia processing, compression, and collaborative intelligence. His  papers have received awards at ICME 2012 and ICIP 2019, and other recognitions (e.g., paper award finalist, top n\%) at Asilomar, ICIP, ICME, and CVPR. He is the vice chair of the IEEE Multimedia Signal Processing Technical Committee and an elected member of the IEEE Multimedia Systems and Applications Technical Committee. He was an Associate Editor of \textsc{IEEE Transactions on Multimedia} and \textsc{IEEE Signal Processing Magazine}, and is currently a Senior Area Editor of \textsc{IEEE Signal Processing Letters}. 
\end{IEEEbiography}







\end{document}